\documentclass[10pt,a4paper]{article}
\usepackage[latin1]{inputenc}
\usepackage[margin=3cm,top=3cm,bottom=3cm]{geometry}

\usepackage{amsmath}
\usepackage{amsthm}
\usepackage{amsfonts}
\usepackage{amssymb}
\usepackage{graphicx}
\usepackage{float}
\usepackage{color}
\usepackage{subfig}
\usepackage{authblk}
\usepackage{natbib}
\usepackage[margin=3cm,top=3cm,bottom=3cm]{geometry}

\newtheorem{proposition}{Proposition}

\newtheorem{example}{Example}

\begin{document}
\title{An analytic multi-currency model with stochastic volatility and stochastic interest rates\thanks{We are particularly grateful to: Carmine Corvasce, Simone Freschi, Tommaso Gabriellini, Christian Fries, Wolfgang Runggaldier.}}

\author{
 \textrm{Alessandro
Gnoatto}\thanks{Mathematisches Institut der LMU, Theresienstrasse, 39 D-80933 M\"unchen, gnoatto@mathematik.uni-muenchen.de, http://www.alessandrognoatto.com.}
\and\textrm{Martino Grasselli}\thanks{Dipartimento di Matematica, Padova (Italy) and D\'{e}partement
Math\'ematiques et Ing\'enierie Financi\`{e}re, ESILV, Paris La
D\'efense (France), and QUANTA FINANZA S.R.L., Via Cappuccina 40, Mestre (Venezia), Italy.} }

\maketitle

\begin{abstract}
We introduce a tractable multi-currency model with stochastic volatility and correlated stochastic interest rates that takes into account the smile in the FX market and the evolution of yield curves. The pricing of vanilla options on FX rates can be performed efficiently through the FFT methodology thanks to the affinity of the model.  Our framework is also able to describe many non trivial links between FX rates and interest rates: a second calibration exercise highlights the ability of the model to fit simultaneously FX implied volatilities while being coherent with interest rate products.
\end{abstract}

\section{Introduction}

 The FX market is the largest and most liquid financial market in the world. The daily volume of FX option transaction in 2010 was
about 207 billion USD, according to \cite{bis2010}. The stylized facts concerning FX options may be ascribed to two main
cathegories: features of the underlying exchange rates and the implied volatilities respectively. The first and most important feature of FX rates is that the inverse
of an FX rate is still an FX rate, so if $S^{d,f}(t)$ is a model for $EUR/USD$ exchange rate, i.e. the price in dollars of one euro, thus reflecting the point
of view of an American investor, then $S^{f,d}(t)=1/S^{d,f}(t)$ represents the $USD/EUR$ rate, i.e. the price in euros of one dollar, hence representing the perspective of a
European investor. This basic observation may be further generalized so as to construct e.g. triangles of currencies where then the no-arbitrage relation $S^{f,d}(t)=S^{f,k}(t)S^{k,d}(t)$ must hold.
This particular property of FX rates must then be coupled with the presence of a volatility smile for each FX rate involved in a currency triangle,
in such a way that we can consider a model who is able to capture jointly relations among underlyings and their respective implied volatilities.\\

Since the financial crisis, investors look for products with a long time horizon that are supposed to be less sensitive to short-term market
 fluctuations. Following \cite{bookClark11}, the risk involved in such structures may be intuitively understood in terms of a simple example. The value of an ATM call option,
 with maturity $T$, in a Black-Scholes setting, is usually approximated by practitioners by means of the formula $0.04\sigma\sqrt{T}$, see \cite{thesis_bac} and \cite{shte08}. Under this
 approximation, the vega of such a position is simply given by $0.04\sqrt{T}$ and hence scales as square root of $T$. If we look instead
 at the rho risk it can be shown that it scales as $T$. This heuristic observation suggests that when we look at shortdated FX options, the volatility
 smile is a dominant factor, whereas, as the maturity increases, the interest rate uncertainty plays an increasingly important role. As a consequence,
 when we consider longdated FX products, a model that is simultaneously able to take into account volatility and interest rate risk is preferable.\\

Besides the simplest case of \textit{currency swap}, longdated FX products received increasing importance since the appearence
of new structured products like power reverse dual currency notes (PRDC). According to \cite{bookClark11}, for some currency pairs ``\textit{exchange rates in the future tend, once
we actually get to those future dates, to be closer to current spot levels than the forward rates observed today would predict}''. PRDC notes try to exploit this
decoupling phenomenon so as to grant a speculative gain to the investor. PRDC notes pay coupons in case the future exchange rates $S^{d,f}(T_k)$ are closer
to the current level of spot $S^{d,f}(t)$ rather than the forward price $F(t,T_k)=S^{d,f}(t)D^d(t,T_k)/D^f(t,T_k)$, where $D^l(t,T), \ l=d,f$ denote the discount factors
under the currencies $f,d$ respectively. Consider a sequence of payment dates $\left\{T_1,...,T_N\right\}$, where typically $T_N$ is 30 years and $\Delta T_k$ is six months.
At these dates the following coupon is paid
\begin{align}
 V_{T_k}=\Delta T_k \max\left\{N^d\left(\frac{S^{d,f}(T_k)}{S^{d,f}(t)}C^f-C^d\right),0\right\},
\end{align}
where $C^f,C^d$ are constants typically equal to respectively $0.15$ and $0.1$. The coupon may be rewritten as
\begin{align}
 V_{T_k}=\Delta T_k\frac{N^dC^f}{S^{d,f}(t)}\left( S^{d,f}(T_k)- K\right)^+,\quad  K=S^{d,f}(t)\frac{C^d}{C^f}.	
\end{align}
The basic example above can be further generalized so as to include capped or floored coupons. Futher examples of longdated FX products are callable,
trigger or chooser PRDC notes. Given the long maturity of these products, it is essential to account for both oscillations in FX and interest rates.\\ 

\subsection{Related literature}
Historically, it has been quite a standard practice to employ, for the pricing of FX options, models originally designed e.g. for equity options, modulo
minor changes. This approach was initiated by \cite{GarmanKohl}, who adapted the \cite{bla73} model, see \cite{bookWystup06}. Another important example is given by the Heston model, whose FX
adaptation is presented e.g. in \cite{bookClark11} and \cite{wystup2010}. Other examples are the Stein and Stein \cite{steinstein91} and the Hull and White \cite{hullWhite87} model, as described e.g. in \cite{book_Lipton}. Another recent contribution is \cite{lewo13} where they adopted the Wishart multifactor stochastic volatility model introduced in \cite{article_DaFonseca1} in a single exchange rate setting. With a view towards quanto options \cite{BrangerMu} also employ the Wishart process.\\

The typical limitation of the previous approaches is that they neglect the relationships among multiple exchange rates. In fact, the joint presence of triangular relationships and volatility smiles makes it difficult to model all FX rates in a triangle of currencies. In the literature,
there exists a stream of contributions that try to recover the risk neutral probability distribution of the cross exchange rate, either by means of joint
densities or copulas. These approaches represent an evolution of the Breeden and Litzenberger \cite{bre78} approach. Among others, we recall \cite{blissPani2002} and \cite{sch12} and \cite{austing2011}. Copulas have been employed in
\cite{bennettKennedy2004}, \cite{salmonSchneider2006} and \cite{hurdSalmonSchleicher2005}. In the presence of a stochastic volatility model of Heston type, \cite{CarrVerma}
try to solve the joint valuation problem of FX options by specifying the dynamics of two rates influenced by a common stochastic volatility factor. This rather restrictive approach
seems however difficult to extend. Asymptotic formulae for a SABR specification are provided in \cite{shirayaTakahashi2012}. The approach we are interested in is
presented in \cite{gnoatto11}, where a multifactor stochastic volatility model of Heston type is introduced. The model is coherent with respect to triangular relationships
among currencies and allows for a simultaneous calibration of the volatility surfaces of the FX rates involved in a triangle, like $EUR/USD/JPY$. The idea of \cite{gnoatto11}
is inspired by the work of \cite{article_heplaten06}, who consider a model for FX rates of the form
\begin{align}
 S^{i,j}(t)=\frac{\mathcal{G}^i(t)}{\mathcal{G}^j(t)}
\end{align}
where $\mathcal{G}^i,\mathcal{G}^j$ represent the value of the growth optimal portfolio under the two currencies. A similar approach, known as intrinsic
currency valuation framework, has been independently proposed in \cite{article_doust} and \cite{doust2012}.\\

In the industry, the evaluation of longdated FX products is usually performed by coupling Hull and White \cite{hul93} models for the short rates in each monetary area 
with a stochastic process for the FX rate, which is usually assumed to be a geometric Brownian motion, see \cite{bookClark11}. Many authors, have studied the problem of combining uncertain FX rates or stocks with stochastic interest rates.
Such a model represents a starting point for the approach presented in \cite{piterbarg06bis}, where a local volatility effect in the FX process is also introduced.
\cite{hlps2009} derive closed-form pricing formulae under the same kind of three factor model, coupled also with a stochastic volatility process of
\cite{shzhu1999} type, see also \cite{ahlip2008} and \cite{ahlip2010}. This kind of setting implies that both interest rates and volatility may become negative. \cite{haapel2009} extend the previous approach and
also consider approximate solutions when the instantaneous variance follows a square root process. \cite{delray2012} first present a local volatility framework, and then extend it to a stochastic volatility setting and provide approximations by means of a Markovian projection.\\

A desirable feature, when we couple interest rates and FX rates with stochastic volatility, is that we would like to use models such that the instantaneous variance and the interest rates remain positive. A natural choice in this sense is given by
the square root process. However, introducing a non-zero correlation among square root processes breaks the analytical tractability of the model, when we work with affine processes on the canonical state space $\mathbb{R}^m_{\geq 0}\times \mathbb{R}^n$.
 \cite{oosterlee11} attack the problem by providing approximations of the non affine terms in the Kolmogorov PDE satisfied by the characteristic function of the forward exchange rate.
   
\subsection{Main results of the article}
In this paper we propose an extension of \cite{gnoatto11} to the case where the stochastic factors driving the volatilities of the exchange rates belong to the matrix state space $S_d^+$, the cone of positive semidefinite $d\times d$ matrices. Stochastically continuous Markov processes on $S_d^+$ with exponential affine dependence on the initial state space have been characterized in \cite{article_Cuchiero}, see also \cite{Mayer_jump}.
We will see that all good analytical properties of the \cite{gnoatto11}  approach are preserved, while we allow for a more general dynamics.
Our model is at the same time an affine multifactor stochastic volatility model for the FX rate, where the instantaneous variance is driven by a Wishart process, see \cite{article_Bru}, \cite{article_DaFonseca1}, and a 
Wishart affine short rate model, see \cite{article_gou02}, \cite{gra08}, \cite{article_BCT}, \cite{article_Chiarella} and \cite{gnoatto12b}. The model has many interesting features, namely:
\begin{itemize}
 \item it can be jointly calibrated on different FX volatility smiles;
 \item it is coherent with triangular relationships among FX rates;
 \item it allows for closed form solutions for both FX options and basic interest rates derivatives;
 \item it allows for non trivial correlations betwen interest rates and volatilities.
\end{itemize}
The first two interesting features are shared with the model in \cite{gnoatto11}. The other results constitute a novel contribution of the present article, which is outlined as follows:
In section \ref{theModel}, we set up our modelling framework and provide an example highlighting the flexibility of the approach. Section \ref{rn_meas} is
devoted to a complete characterization of all risk-neutral measures associated to the different economies while Section \ref{features} shows that the present setting
allows for stochastic correlations among many economic quantities. Section \ref{pricing} presents all closed form pricing formulas for FX and interest rates derivatives, together with
asymptitic expansions of FX implied volatilities. The numerical treatment of our model is the topic of Section \ref{pricing}, where  fit the model to an FX surface and two yield curves. Finally Section \ref{conclusions} summarizes our findings. Technical proofs are gathered in the Appendix.

\section{The model}\label{theModel}

We consider a foreign exchange market in which $N$ currencies are traded between
each other via standard FX spot and FX vanilla option transactions.  The value of each
of these currencies $i=1,..,N$ in units of a universal num\'{e}raire is denoted by  $S^{0,i}(t)$ (note that
$S^{0,i}(t)$ can itself be thought as an exchange rate, between the currency $i$ and the artificial currency $i=0$: we will see that the results are independent of the specification of the universal num�raire).

We assume the existence of $N$ money-market accounts (one for each monetary area), whose values are driven by locally deterministic ODE's of the type:
\begin{align}
dB^i(t)=&r^i(t)B^i(t)dt,\quad \quad i=1,..,N;
\end{align}
and we denote with $r^0$ the interest rate corresponding to the artificial currency.

We assume the existence of a common multivariate stochastic factor $\Sigma$ driving both the interest rates and the volatilities of the exchange rates $S^{0,i}(t)$. We model the stochastic factor  $\Sigma$ as a matrix Wishart process (see \cite{article_Bru}) evolving as 
\begin{align}
   d\Sigma(t)&=(\Omega\Omega^{\top}+M\Sigma(t)+\Sigma(t)M^{\top})dt+\sqrt{\Sigma(t)}dW(t)Q+Q^{\top}dW(t)^{\top}.\label{dyn_wis}
\end{align}
We assume $\Omega,M,Q\in M_{d}$, the set of $d\times d$ real matrices, $W=\left(W_{t}\right)_{t\geq0}\in M_d$ is a matrix Brownian motion (i.e. a $d\times d$ matrix whose components are independent Brownian motions).  The dimension $d$ can be chosen according to the specific problem and may reflect a PCA-type analysis. In order to ensure the typical mean reverting behavior of the process we assume that $M$ is negative semi-definite, moreover the matrix $\Omega$ satisfies the condition $\Omega\Omega^{\top}=\beta Q^{\top}Q$, for $\beta\geq d+1$. This last condition ensures the existence of a unique strong solution to the SDE \eqref{dyn_wis}, according to Corollary 3.2 in \cite{article_MPS}.

We model each of the $S^{0,i}(t)$ via a multifactor Wishart
stochastic volatility model. Formally, we write
\begin{align}
    \frac{dS^{0,i}(t)}{S^{0,i}(t)}  &= (r^0(t)-r^i(t)) dt-Tr\left[A_i\sqrt{\Sigma(t)}dZ(t)\right],\quad \ &i=1,\dots,N; \label{rates2}
\end{align}
where $Z_t\in M_d$ denotes a matrix Brownian motion.

 The diffusion term exhibits a structure that is completely analogous to the one introduced in \cite{article_heplaten06} and \cite{gnoatto11}: in the present case we have that the dynamics of the exchange rate is driven by a linear projection of the variance factor $\sqrt{\Sigma(t)}$ along a direction parametrized by the symmetric matrix $A_i$. As a consequence the total instantaneous variance is $Tr\left[A_i\Sigma(t)A_i\right]dt$.

The stochastic factor $\Sigma$ drives also the short interest rates:
\begin{align}
r^0&=h^0+Tr\left[H^0\Sigma(t)\right]\\
r^i&=h^i+Tr\left[H^i\Sigma(t)\right],\label{shortRates}
\end{align}
for $h^k>0$, $H^k\in S_d^+$, $k=0,..,N$.

We assume a correlation structure between the two matrix Brownian motions $Z(t)$ and $W(t)$, by means of an invertible matrix $R$ according to the following relationship:
\begin{align}
W(t)=&Z(t)R^\top +B(t)\sqrt{I_d-RR^\top},\label{corrStruct}
\end{align}
where $B(t)$ is a matrix Brownian motion independent of $Z(t)$.
We denote by $S^{i,j}(t)$, $i,j=1,..,N$ the exchange rate between currency $i$ and $j$. By Ito's lemma we have that $S^{i,j}(t)= S^{0,j}(t)/S^{0,i}(t)$ has the following dynamics
\begin{align}
\frac{dS^{i,j}(t)}{S^{i,j}(t)}&=(r^i(t)-r^j(t))dt+Tr[(A_i -A_j)\Sigma(t)A_i]dt\nonumber\\
&+Tr[(A_i -A_j)\sqrt{\Sigma(t)}dZ(t)].\label{Wrates}
\end{align}
The additional drift term in (\ref{Wrates}) can be understood as a quanto adjustment between the currencies $i$ and $j$.\\

 We try to provide some intuition concerning the flexibility of the approach by considering an introductory example.

\begin{example}
Let $d=4$. Consider the case
\begin{align*}
r^i&=h^i+Tr\left[\left(
\begin{array}{cccc}
1&0&0&0\\
0&0&0&0\\
0&0&0&0\\
0&0&0&0	
\end{array}
\right)\Sigma(t)\right]=h^i+\Sigma_{11}(t)\\
r^j&=h^j+Tr\left[\left(
\begin{array}{cccc}
0&0&0&0\\
0&1&0&0\\
0&0&0&0\\
0&0&0&0	
\end{array}
\right)\Sigma(t)\right]=h^j+\Sigma_{22}(t)\\
\end{align*}
Moreover we let the matrices $A^i,A^j$ be partitioned as follows
\begin{align*}
A^i=\left(
\begin{array}{cccc}
0&0&0&0\\
0&0&0&0\\
0&0&A^i_{11}&A^i_{12}\\
0&0&A^i_{12}&A^i_{22}	
\end{array}
\right)	
\end{align*}

The idea is that the first two diagonal elements of the Wishart process are mainly responsible for the dynamics of the short-rates whereas the third and the fourth drive the volatilities. The elements of the Wishart process are related among each other in a non-trivial way, in fact
\begin{align}
d\left\langle \Sigma_{ij},\Sigma_{jk}\right\rangle_t=\left(\Sigma_{ik}\left(Q^\top Q\right)_{jl}+\Sigma_{il}\left(Q^\top Q\right)_{jk}+\Sigma_{jl}\left(Q^\top Q\right)_{ik}\right)dt.
\end{align}
Consequently, in the present four-dimensional hybrid FX-short-rate model, the correlation between e.g. the element $\Sigma_{11}$ (which is mainly responsible for the dynamics of the first short rate) and $\Sigma_{33}$ (which drives the volatility of the FX rate) is stochastic and given by
\begin{align*}
d\left\langle \Sigma_{11},\Sigma_{33}\right\rangle_t=4\Sigma_{13}\left(Q^\top Q\right)_{13}dt.
\end{align*}

From the discussion above, we realize that we are in presence of a hybrid model, allowing for non trivial links among interest rates and FX rates, while preserving full analytical tractability, as we will see in the sequel. This interesting feature, to the best of our knowledge, is not shared by other existing hybrid models, that usually rely on approximations of the solution of the Kolmogorov PDE involving non-affine terms (see e.g. \cite{oosterlee11}, \cite{oosterlee12}).
\end{example}

\section{Risk neutral probability measures}\label{rn_meas}
Up to now we have worked in the risk neutral measure defined by our artificial currency.
In practical pricing applications, it is more convenient to change the num\'{e}raire to any of the currencies included in our
FX multi-dimensional system. Without loss of generality, let us consider the risk neutral measure defined by the $i$-th
money market account $B^i(t)$ and derive the dynamical equations for the standard FX rate $S^{i,j}(t)$, its inverse $S^{j,i}(t)$, and a generic cross $S^{j,l}(t)$.

The Girsanov change of measure that transfers to the $\mathbb{Q}^i$ risk neutral measure (i.e. the risk neutral measure in the $i$-th country) is simply determined by assuming that under $\mathbb{Q}^i$ the drift of the exchange rate $S^{i,j}(t)$ is given by $r^i-r^j$ (or equivalently by the fact that the money market account $B^f(t)$ is a $\mathbb{Q}^i$-martingale once discounted by the interest rate $r^i(t)$). The associated Radon-Nikodym derivative is given by
\begin{align*}
\left.\frac{d\mathbb{Q}^i}{d\mathbb{Q}^0}\right|_{\mathcal{F}_t}=& \exp\left(-\int_0^t Tr[A_i\sqrt{\Sigma(s)}dZ(s)]-\frac{1}{2}\int_0^t Tr[A_i\Sigma(s)A_i] ds\right).
\end{align*}
In \cite{Mayer2012}, conditions ensuring that the stochastic exponential above is a true martingale are provided.
In the following, we proceed along the lines of \cite{gnoatto11}. The possibility of buying the foreign currency and investing it at the foreign short rate of interest, is equivalent to the possibility of investing in a domestic asset with price process $\tilde{B}^i_j=B^jS^{i,j}$, where $i$ is the domestic economy and $j$ is the foreign one. Then
\begin{align*}
d\tilde{B}_{j}^{i}(t)&=d\left(B^{j}(t)S^{i,j}(t)\right)\nonumber\\
&=B^{j}(t)S^{i,j}(t)\left((r^i(t)-r^j(t))dt+Tr[(A_i -A_j)\Sigma(t)A_i]dt\right.\nonumber\\
&\left.+Tr[(A_i -A_j)\sqrt{\Sigma(t)}dZ(t)]\right)+B^{j}S^{i,j}(t)r^{j}(t)dt\nonumber\\
&=\tilde{B}_{j}^{i}(t)\left(r^i(t)dt+Tr[(A_i -A_j)\Sigma(t)A_i]dt+Tr[(A_i -A_j)\sqrt{\Sigma(t)}dZ(t)]\right)\nonumber\\
&=\tilde{B}_{j}^{i}(t)\left(r^{i}(t)dt+Tr[(A_i -A_j)\sqrt{\Sigma(t)}dZ^{\mathbb{Q}^i}(t)]\right),
\end{align*}
where the matrix brownian motion under $\mathbb{Q}^i$ is given by
\begin{align*}
dZ^{\mathbb{Q}^i}=dZ+\sqrt{\Sigma(t)}A_idt,
\end{align*}
then the $\mathbb{Q}^i$-risk neutral dynamics of the exchange rate is of the form
\begin{align*}
dS^{i,j}(t)&=d\left(\frac{\tilde{B}_{j}^{i}(t)}{B_{j}(t)}\right)\nonumber\\
&=S^{i,j}(t)\left((r^{i}(t)-r^{j}(t))dt+Tr\left[(A_i -A_j)\sqrt{\Sigma(t)}dZ^{\mathbb{Q}^i}(t)\right]\right).
\end{align*}
The measure change has however also an impact on the variance processes, via the correlation matrix $R$ introduced in \eqref{corrStruct}.
The component of $dB(t)$ that is orthogonal to the spot driver $dZ(t)$ is not affected by the measure
change; this is a natural choice that is consistent with the foreign-domestic symmetry\footnote{This assumption is in line with the procedure that has been introduced in \cite{gnoatto11} and implies that the model is consistent with the foreign-domestic parity as in \cite{note_dbr}.}.
 We are now able to derive the risk neutral dynamics of the factor process $\Sigma(t)$ governing the volatility of the exchange rates under $\mathbb{Q}^i$, that is given by
\begin{equation}
dW^{\mathbb{Q}^i}(t)=\left(dZ(t)+\sqrt{\Sigma}A_idt\right)R^{\top}+dB(t)\sqrt{I_d-RR^{\top}}.\label{bm_Qi}
\end{equation}
From \eqref{dyn_wis} and \eqref{corrStruct} we derive the $\mathbb{Q}^i$-risk neutral dynamics of $\Sigma$ as follows:
\begin{align*}
d\Sigma(t)=&(\Omega\Omega^{\top}+M\Sigma(t)+\Sigma(t)M^{\top})dt\nonumber\\
&+\sqrt{\Sigma(t)}\left(\left(dZ(t)+\sqrt{\Sigma(t)}A_idt\right)R^{\top}+dB(t)\sqrt{I_d-RR^{\top}}\right)Q\nonumber\\
&+Q^{\top}\left(R\left(dZ^{\top}_t+A_i\sqrt{\Sigma(t)}dt\right)+\sqrt{I_d-RR^{\top}}^{\top}dB^{\top}\right)\sqrt{\Sigma(t)}\nonumber\\
&-\Sigma(t)A_iR^{\top}Qdt-Q^{\top}RA_i\Sigma(t)dt.\nonumber
\end{align*}
Now define
\begin{align*}
M^{\mathbb{Q}^i}:=M-Q^{\top}RA_i,
\end{align*}
so that using \eqref{bm_Qi} we can finally write
\begin{align*}
d\Sigma(t)&=(\Omega\Omega^{\top}+M^{\mathbb{Q}^i}\Sigma(t)+\Sigma(t)M^{\mathbb{Q}^i,\top})dt\nonumber\\
&\sqrt{\Sigma(t)}dW^{\mathbb{Q}^i}(t)Q+Q^{\top}dW^{\mathbb{Q}^i,\top}(t)\sqrt{\Sigma(t)},
\end{align*}
from which we deduce the relations among the parameters:
\begin{align}
R^{\mathbb{Q}^i}&=R,\nonumber\\
Q^{\mathbb{Q}^i}&=Q,\\
M^{\mathbb{Q}^i}&=M-Q^{\top}RA_i\nonumber.
\end{align}
We observe that, like in the multi-Heston case of \cite{gnoatto11}, the functional form of the model is invariant under the measure change between $\mathbb{Q}^0$ and the $i$th-risk neutral measure. The inverse FX rate under the $\mathbb{Q}^i$-risk neutral measure follows from Ito calculus, recalling that $S^{j,i}=\left(S^{i,j}\right)^{-1}$:
\begin{align*}
    \frac{dS^{j,i}(t)}{S^{j,i}(t)} & = S^{i,j}(t) d\left(\frac{1}{S^{i,j}(t)}\right)  \\
    & =\left(r^j(t)-r^i(t) + Tr\left[(A_j - A_i)\Sigma(t)(A_j - A_i)\right]\right)dt\nonumber\\
    &+Tr\left[(A_j-A_i)\sqrt{\Sigma(t)}dZ^{\mathbb{Q}^i}(t)\right],\nonumber
\end{align*}
which includes the self-quanto adjustment. Similarly, the SDE of a generic cross FX rate becomes
\begin{align*}
    \frac{dS^{j,l}(t)}{S^{j,l}(t)}= &  \frac{S^{i,j}(t)}{S^{i,l}(t)}\,d\left(\frac{S^{i,l}(t)}{S^{i,j}(t)}\right)   \\
     =&  \left(r^j(t)-r^l(t) + Tr\left[(A_j- A_l)\Sigma(t)(A_j - A_i)\right]\right)dt\nonumber\\
    &+Tr\left[(A_j-A_l)\sqrt{\Sigma(t)}dZ^{\mathbb{Q}^i}(t)\right].\nonumber
\end{align*}
The additional drift term is the quanto adjustment as described by the current model choice. By applying
Girsanov's theorem again, this time switching to the $\mathbb{Q}^j$ risk neutral measure, the term is removed
while the Wishart parameters change according to the following fundamental transformation rules:

\begin{align}
    R^{\mathbb{Q}^j}&=R^{\mathbb{Q}^i},\nonumber\\
    Q^{\mathbb{Q}^j}&=Q^{\mathbb{Q}^i},\\
    M^{\mathbb{Q}^j}&=M^{\mathbb{Q}^i}-Q^{\mathbb{Q}^i,\top}R^{\mathbb{Q}^i}(A_j-A_i)\nonumber.
\end{align}

\section{Features of the model}\label{features}

\subsection{Functional symmetry of the model}

Recall that a crucial property of the FX market requires that products or ratios of exchange rates in a triangle are also exchange rates, meaning that the dynamics of the exchange rates must be functional symmetric with respect to which FX pairs
we choose to be the main ones and which one the cross. That is, it is not
a priori trivial to obtain a model such that the dynamics for the inverse of an exchange rate shares the same functional form in its coefficients. Or equivalently, the dynamics of $(S^{i,l}S^{l,j})$ computed
by applying
the Ito's rule to the product $(S^{0,l}/S^{0,i})\times (S^{0,j}/S^{0,l})$, must give the dynamics
of a process that shares the same functional form of both $S^{i,l}$ and $S^{l,j}$. This
symmetry property is fundamental in order to be able to joint calibrate and
price consistently multi currency options (see e.g. \cite{gnoatto11}).

\begin{proposition}\label{prop_6}
The dynamics of the exchange rates (\ref{Wrates}) satisfies the triangular relation, namely the model is functional symmetric.
\end{proposition}
\begin{proof} See the Appendix.
\end{proof}

\subsection{Stochastic Skew}
In analogy with \cite{gnoatto11}, if we calculate
the infinitesimal correlation between the log returns of $S^{i,j}$ and their variance
$Vol^2(S^{i,j} )$, we find that it is stochastic. This is a nice feature of the model since it implies that the skewness of vanilla options on $S^{i,j}$ is stochastic, which is a well known stylized fact in the FX market (see e.g. \cite{carrwu07}). 
%
%
%

Let us consider the infinitesimal variance of $S^{i,j}$:
\begin{equation*}
d\left\langle S^{i,j},S^{i,j}\right\rangle_t=Tr[(A_i-A_j)\Sigma(t)(A_i-A_j)]dt,
\end{equation*}
so that we can write 
\begin{align}
\frac{dS^{i,j}(t)}{S^{i,j}(t)}&=(r^{i}-r^{j})dt+Tr\left[(A_i -A_j)\sqrt{\Sigma(t)}dZ^{\mathbb{Q}^i}(t)\right]\nonumber\\
&=(r^{i}-r^{j})dt+\sqrt{Tr\left[(A_i -A_j)\Sigma(t)(A_i -A_j)\right]}dB_1(t),\label{weak}
\end{align}
where we defined
\begin{align*}
dB_1(t):=\frac{Tr\left[(A_i -A_j)\sqrt{\Sigma(t)}dZ^{\mathbb{Q}^i}(t)\right]}{\sqrt{Tr\left[(A_i -A_j)\Sigma(t)(A_i -A_j)\right]}}.
\end{align*}
The process $B_1=\left(B_1(t)\right)_{t\geq 0}$ is a scalar local martingale with quadratic variation given by $t$, hence an application of L\'evy characterization theorem allows us to claim that $B_1$ is a  Brownian motion and we still denote by $S^{i,j}$ the weak solution to the SDE \eqref{weak} .

The dynamics of the variance is given by:
\begin{align*}
&dTr\left[(A_i -A_j)\Sigma(t)(A_i -A_j)\right]\nonumber\\
&=\left(Tr\left[(A_i -A_j)\Omega\Omega^{\top}(A_i -A_j)\right]+2Tr\left[(A_i -A_j)M\Sigma(t)(A_i -A_j)\right]\right)dt\nonumber\\
&+2Tr\left[(A_i -A_j)\sqrt{\Sigma}dW^{\mathbb{Q}^i}(t)Q(A_i -A_j)\right].
\end{align*}
In order to determine the scalar Brownian motion driving the variance process we shall compute the following quadratic variation:
\begin{align*}
&d\left\langle Tr\left[(A_i -A_j)\Sigma(A_i -A_j)\right]\right\rangle_t\nonumber\\
&=4d\left\langle\int_0^. Tr\left[(A_i -A_j)\sqrt{\Sigma(v)}dW^{\mathbb{Q}^i}(v)Q(A_i -A_j)\right],\right.\nonumber\\
&\left. \int_0^. Tr\left[(A_i -A_j)\sqrt{\Sigma(u)}dW^{\mathbb{Q}^i}(u)Q(A_i -A_j)\right] \right\rangle_t\nonumber\\
&=4d\left\langle \int_0^.\sum_{a,b,c,d,e=1}^{d}{(A_i -A_j)_{ab}\sqrt{\Sigma(v)}_{bc}dW_{cd}^{\mathbb{Q}^i}(v)Q_{de}(A_i -A_j)_{ea}},\right.\nonumber\\
&\left.\int_0^.\sum_{p,q,r,s,l=1}^{d}{(A_i -A_j)_{pq}\sqrt{\Sigma(u)}_{qr}dW_{rs}^{\mathbb{Q}^i}(u)Q_{sl}(A_i -A_j)_{lp}}\right\rangle_t\nonumber\\
&=\sum_{a,b,e,p,q,r,s,t=1}^{d}(A_i -A_j)_{ab}\sqrt{\Sigma(t)}_{br}Q_{se}(A_i -A_j)_{ea}\nonumber\\
&\times(A_i -A_j)_{pq}\sqrt{\Sigma(t)}_{qr}Q_{sl}(A_i -A_j)_{lp}dt\nonumber\\
&=4Tr\left[(A_i -A_j)^2\Sigma(t)(A_i -A_j)^2Q^{\top}Q\right]dt.
\end{align*}
Then we can use the same arguments as before and express the dynamics of the variance as follows:
\begin{align}
&dTr\left[(A_i -A_j)\Sigma(t)(A_i -A_j)\right]\nonumber\\
&=\left(...\right)dt+2\sqrt{Tr\left[(A_i -A_j)^2\Sigma(t)(A_i -A_j)^2Q^{\top}Q\right]}dB_2(t),\nonumber
\end{align}
where
\begin{align*}
dB_2(t):=\frac{Tr\left[(A_i -A_j)\sqrt{\Sigma}dW^{\mathbb{Q}^i}(t)Q(A_i -A_j)\right]}{\sqrt{Tr\left[(A_i -A_j)^2\Sigma(t)(A_i -A_j)^2Q^{\top}Q\right]}},
\end{align*}
which allows us to compute the covariation between the two noises. Notice that in the calculation above we are assuming the invariance of the correlation with respect to the change of measure that was explained in Section \ref{rn_meas}. The skewness is then related to the quadratic covariation between the two noises $B_1,B_2$: 
\begin{align}
d\left\langle B_1,B_2 \right\rangle_t
&=\frac{d\left\langle \int_0^. Tr\left[(A_i -A_j)\sqrt{\Sigma(t)}dZ_{t}^{\mathbb{Q}^i}\right],
\int_0^. Tr\left[(A_i -A_j)\sqrt{\Sigma(t)}dW^{\mathbb{Q}^i}(t)Q(A_i -A_j)\right]\right\rangle_t}{\sqrt{Tr\left[(A_i -A_j)\Sigma(t)(A_i -A_j)\right]}\sqrt{Tr\left[(A_i -A_j)^2\Sigma(t)(A_i -A_j)^2Q^{\top}Q\right]}} \nonumber\\
&=\frac{d\left\langle \int_0^.Tr\left[(A_i -A_j)\sqrt{\Sigma(t)}dZ^{\mathbb{Q}^i}(t)\right],
\int_0^.Tr\left[(A_i -A_j)\sqrt{\Sigma(t)}dZ^{\mathbb{Q}^i}(t)R^{\top}Q(A_i -A_j)\right]\right\rangle_t}
{\sqrt{Tr\left[(A_i -A_j)\Sigma(t)(A_i -A_j)\right]}\sqrt{Tr\left[(A_i -A_j)^2\Sigma(t)(A_i -A_j)^2Q^{\top}Q\right]}}\nonumber\\
&=\frac{\sum_{a,b,c,d,e,f,p,q,r=1}^{d}{(A_i -A_j)_{pq}\sqrt{\Sigma(t)}_{qr}dZ_{rp}^{\mathbb{Q}^i}(t)(A_i -A_j)_{ab}\sqrt{\Sigma(t)}_{bc}dZ_{cd}^{\mathbb{Q}^i}(t)R_{de}^{\top}Q_{ef}(A_i -A_j)_{fa}}}{\sqrt{Tr\left[(A_i -A_j)\Sigma(t)(A_i -A_j)\right]}\sqrt{Tr\left[(A_i -A_j)^2\Sigma(t)(A_i -A_j)^2Q^{\top}Q\right]}}\nonumber\\
&=\frac{\sum_{a,b,c,e,f,q=1}^{d}{(A_i -A_j)_{dq}\sqrt{\Sigma(t)}_{qc}\sqrt{\Sigma(t)}_{cb}(A_i -A_j)_{ba}(A_i -A_j)_{af}Q_{fe}^{\top}R_{ed}}}{\sqrt{Tr\left[(A_i -A_j)\Sigma(t)(A_i -A_j)\right]}\sqrt{Tr\left[(A_i -A_j)^2\Sigma(t)(A_i -A_j)^2Q^{\top}Q\right]}}dt\nonumber\\
&=\frac{Tr\left[(A_i -A_j)\Sigma(t)(A_i -A_j)^{2}Q^{\top}R\right]}{\sqrt{Tr\left[(A_i -A_j)\Sigma(t)(A_i -A_j)\right]}\sqrt{Tr\left[(A_i -A_j)^2\Sigma(t)(A_i -A_j)^2Q^{\top}Q\right]}}dt.
\end{align}
This quantity is proportional to the skew: in particular, by looking at the numerator we realize the proportionality in the asymptotic expansion of the Proposition \ref{exp_Wis1}.

\subsection{A stochastic variance-covariance matrix}
We would like to discuss the positive definiteness of the variance-covariance matrix. For simplicity, we consider the case of three currencies, meaning that we will have a $2\times 2$ candidate covariance matrix:
\begin{align}
\left(\begin{array}{cc}
\left\langle \ln S^{i,j} \right\rangle_t&\left\langle \ln S^{i,j}, \ln S^{i,l}\right\rangle_t\\
\left\langle \ln S^{i,j},\ln S^{i,l}\right\rangle_t&\left\langle \ln S^{i,l} \right\rangle_t 	
\end{array}
\right).
\end{align}
We know that
\begin{align}
d\left\langle \ln S^{i,j} \right\rangle_t&=Tr\left[\left(A_i-A_j\right)\Sigma(t)\left(A_i-A_j\right)\right]dt,\label{variance2}\\
d\left\langle \ln  S^{i,j},lnS^{i,l}\right\rangle_t&=Tr\left[\left(A_i-A_j\right)\Sigma(t)\left(A_i-A_l\right)\right]dt.
\end{align}
We first look at \eqref{variance2}. We recall that we assumed $A_i,A_j,A_l\in S_d$. Without loss of generality (otherwise put $V=-V'$ for $V'\in S_d^+$), let $\left(A_i-A_j\right)\in S_d^+$. Recall that the cone $S_d^+$ is self dual, meaning that:
\begin{align*}
S_d^+=\left\{u \in S_d\quad|\quad Tr\left[uv\right]\geq0, \forall v\in S_d^+\right\}.
\end{align*}
Let $O$ be an orthogonal matrix, then we may write: $\left(A_i-A_j\right)=O\Lambda O^{\top}$, where $\Lambda$ is a diagonal matrix containing the eigenvalues of $\left(A_i-A_j\right)$ on the main diagonal. Then we have:
\begin{align*}
Tr\left[\left(A_i-A_j\right)\Sigma(t)\left(A_i-A_j\right)\right]&=Tr\left[O\Lambda O^{\top}\Sigma(t) O\Lambda O^{\top}\right]\nonumber\\
&=Tr\left[\Sigma(t) O\Lambda^2 O^{\top}\right]\geq 0
\end{align*}
by self-duality. This shows that variances are positive. Now we would like to check that the variance-covariance matrix is in $S_d^+$.
Now let
\begin{align*}
\mathcal{M}(t)&=\left(A_i-A_j\right)\sqrt{\Sigma(t)},\\
\mathcal{N}(t)&=\left(A_i-A_l\right)\sqrt{\Sigma(t)},
\end{align*}
then, using the Cauchy-Schwarz inequality for matrices we have
\begin{align*}
&Tr\left[\left(A_i-A_j\right)\Sigma(t)\left(A_i-A_j\right)\right]Tr\left[\left(A_i-A_l\right)\Sigma(t)\left(A_i-A_l\right)\right]\\
&=Tr\left[\mathcal{M}(t)\mathcal{M}^{\top}(t)\right]Tr\left[\mathcal{N}(t)\mathcal{N}^{\top}(t)\right]\\
&\geq Tr\left[\mathcal{M}(t)\mathcal{N}^{\top}(t)\right]^2\\
&=Tr\left[\left(A_i-A_j\right)\Sigma(t)\left(A_i-A_l\right)\right]^2.
\end{align*}
This implies that the determinant of the instantaneous variance-covariance matrix is positive, so we conclude that the variance-covariance matrix is well defined, and, as a side effect, we have the usual bound for the correlations, i.e. all correlations are bounded by one (in absolute value).

\subsection{Stochastic correlation between short rates and FX rates}

We now show that our model allows for non trivial dependencies between the short rates and the exchange rates.

\begin{proposition}
The instantaneous correlation between a short rate $r^i$ and the log-FX rate $\log S^{ij}(t)$ is stochastic and given by
\begin{align}
\rho_{r^i,\log S^{i,j}(t)}=\frac{Tr\left[\left(A_i-A_j\right)\Sigma(t)HQ^\top R\right]}{\sqrt{Tr\left[\left(A_i-A_j\right)\Sigma(t)\left(A_i-A_j\right)\right]}\sqrt{Tr\left[QH^i\Sigma(t)H^iQ^\top\right]}}.
\end{align}
\end{proposition}

\begin{proof}
Under the $\mathbb{Q}^i$-risk neutral measure, the covariation between the short rate and the log-FX rate is given by

\begin{align*}
d\left\langle \log S^{ij}, r^i\right\rangle_t
&=d\left\langle \int_0^.Tr\left[\left(A_i-A_j\right)\sqrt{\Sigma(u)}dZ(u)\right],\int_0^.2Tr\left[H\sqrt{\Sigma(u)}dZ(u)R^\top Q\right]\right\rangle_t\\
&=2d\left\langle\sum_{i,j,k=1}^{d}\int_0^.\left(A_i-A_j\right)_{ij}\sigma_{jk}(u)dZ_{ki}(u),\int_0^.\sum_{p,q,r,s,n=1}^{d}H_{pq}\sigma_{qr}(u)dZ_{rs}(u)R_{ns}Q_{np}\right\rangle_t\\
&=\left(A_i-A_j\right)_{sj}\sigma_{jr}\sigma_{rq}H_{qp}Q_{np}R_{ns}dt\\
&=Tr\left[\left(A_i-A_j\right)\Sigma(t)HQ^\top R\right]dt,
\end{align*}
while for the short rate we have
\begin{align}
d\left\langle r^i,r^i\right\rangle_t=Tr\left[QH^i\Sigma(t)H^iQ^\top\right]dt.
\end{align}

Given the instantaneous quadratic variation of the log-FX rate and the short rate we conclude.\end{proof}

\subsection{Stochastic correlation between short rates and the variance of the FX rates}

The richness of our model specification may be further appreciated when we look at the correlation between any of the short rates and the variance of the FX rates, which is not usually captured in the literature (see \cite{oosterlee12}).

\begin{proposition}
The instantaneous correlation between the short rate $r^i$ and the variance of the FX rate is stochastic and given by

\begin{align}
\rho_{r^i,Var^{i,j}}=\frac{Tr\left[\left(A_i-A_j\right)\Sigma(t)H^{i}Q^\top Q\left(A_i-A_j\right)\right]}{\sqrt{Tr\left[\left(A_i-A_j\right)^2\Sigma(t)\left(A_i-A_j\right)^2Q^\top Q\right]}\sqrt{Tr\left[QH^i\Sigma(t)H^{i}Q^\top\right]}}.
\end{align}
\end{proposition}

\begin{proof}
Recall that the noise of the scalar instantaneous variance process is
\begin{align*}
dTr\left[\left(A_i-A_j\right)\Sigma(t)\left(A_i-A_j\right)\right]&=\left(...\right)dt+2Tr\left[\left(A_i-A_j\right)\sqrt{\Sigma(t)}dWQ\left(A_i-A_j\right)\right],
\end{align*}
and observe that the noise of the short rate is simply given by
\begin{align*}
dr^i(t)=\left(...\right)dt+2Tr\left[H^i\sqrt{\Sigma(t)}dW(t)Q\right],
\end{align*}
consequently
\begin{align*}
&d\left\langle r^i,Tr\left[\left(A_i-A_j\right)\Sigma\left(A_i-A_j\right)\right]\right\rangle_t\nonumber\\
&=d\left\langle \int_0^. 2Tr\left[H^i\sqrt{\Sigma(u)}dW(u)Q\right],\int_0^. 2Tr\left[\left(A_i-A_j\right)\sqrt{\Sigma(u)}dW(u)Q\left(A_i-A_j\right)\right]\right\rangle_t\nonumber\\
&=d\left\langle\int_0^.\sum_{i,j,k,l=1}^dH_{ij}\sigma_{jk}(u)dW_{kl}(u)Q_{li},\int_0^.\sum_{p,q,r,s,n=1}^d\left(A_i-A_j\right)_{pq}\sigma_{qr}(u)dW_{rs}(u)Q_{sn}\left(A_i-A_j\right)_{np}\right\rangle_t\nonumber\\
&=4\sum_{i,j,p,q,r,s,n=1}H_{ij}\sigma_{jr}(t)Q_{si}\left(A_i-A_j\right)_{pq}\sigma_{qr}(t)dW_{rs}(t)Q_{sn}\left(A_i-A_j\right)_{tp}dt\nonumber\\
&=4\sum_{i,j,p,q,r,s,n=1}\left(A_i-A_j\right)_{pq}\sigma_{qr}(t)\sigma_{rj}(t)H_{ji}Q^\top_{is}Q_{sn}\left(A_i-A_j\right)_{np}dt\nonumber\\
&=Tr\left[\left(A_i-A_j\right)\Sigma(t)HQ^\top Q\left(A_i-A_j\right)\right]dt.
\end{align*}

Combining the above results with the quadratic variation of the instantaneous variance and the quadratic variation of the short rate we conclude.
\end{proof}

\section{Pricing of derivatives}\label{pricing}

A distinctive feature of our model is the ability to price in closed form derivatives written on different underlyings, meaning that we can consider, in a unified approach, more markets simultaneously. In particular, we can consider jointly the FX and the fixed-income market. Basic European products like calls on FX rates and  interest rate products related to different monetary areas can be analyzed together in a single model. In principle, this feature allows the desk to be jointly fitted to different markets by means of a single model. In view of this, we provide fast pricing formulas in semi-closed form, up to Fourier integrals.

\subsection{European FX options}
We first provide the calculation of the discounted Laplace transform and the characteristic function of $x^{i,j}(t):=\ln S^{i,j}(t)$, that will be useful for option pricing purposes. 
Let us consider a call option $C(S^{i,j}(t),K^{i,j},\tau), i,j=1,..,N,i\not=j,$ on a generic FX rate $S^{i,j}(t) = \exp(x^{i,j}(t))$ with
strike $K^{i,j}$, maturity $T$ ($\tau = T - t$ is the time to maturity)  and face equal to one unit of the foreign currency.
For ease of notation set: $R^{\mathbb{Q}^i}=R$ and $Q^{\mathbb{Q}^i}=Q$ and we will use the shorthand $M^{\mathbb{Q}^i}=\tilde{M}$. We proceed to prove the following:
Being an affine model, the discounted characteristic function conditioned
on the initial values
\begin{align}
    \phi^{i,j}(\omega,t, \tau, x ,\Sigma) =  \mathbb{E}^{\mathbb{Q}^i}_t[e^{-\int_t^t r^i_sds}e^{\mathtt{i}\omega x^{i,j}(T)}|x^{i,j}(t) = x,\Sigma(t) = \Sigma  ]
\end{align}
can be derived analytically (here $\mathtt{i}=\sqrt{-1}$).  Standard numerical integration methods can then be used to invert
the Fourier transform to obtain the probability density at $T$ or the vanilla price via integration
against the payoff, with overall little computational effort. In fact, from the usual risk-neutral argument, the initial price of the call option can be written as (domestic) risk neutral expected value:
\begin{equation*}
C(S^{i,j}(t),K^{i,j},\tau)=\mathbb{E}^{\mathbb{Q}^i}_t\left[e^{-\int_t^Tr^i_s ds} \left( e^{x^{i,j}(T)}-K^{i,j}\right)^+\right] ,
\end{equation*}%
and by applying standard arguments (see e.g. \cite{article_Carr99} and also \cite{bak00}, \cite{DPS} and \cite{sepp03}) it can be expressed in terms of the integral of the product of the Fourier transform of the payoff and the discounted characteristic function of the log-asset price:
\begin{align}
C(S^{i,j}(t),K^{i,j},\tau)& =\frac{1}{ 2\pi }\int_{\mathcal{Z}}\phi^{i,j}(-\mathtt{i}\lambda,t, \tau, x ,\Sigma){\Phi}(\lambda)d\lambda,  \label{price2}
\end{align}%
where 
\begin{equation*}
\Phi(\lambda)=\int_{\mathcal{Z}}e^{\mathtt{i} \lambda x}\left( e^{x}-K^{i,j}\right)^+dx
\end{equation*}
is the Fourier transform of the payoff function and $\mathcal{Z}$ denotes the strip of regularity of the payoff, that is the admissible domain where the integral in (\ref{price2}) is well defined. In other words, the pricing problem is essentially solved once the (conditional) discounted characteristic function of the log-exchange rate is known. We recall the relationship between the characteristic function and the moment generating function. If we denote via $G^{i,j}(\omega,t, \tau, x ,\Sigma)$ the discounted moment generating function, given by
\begin{align}
    G^{i,j}(\omega,t, \tau, x ,\Sigma) =  \mathbb{E}^{\mathbb{Q}^i}_t[e^{-\int_t^Tr^i_s ds}e^{\omega x^{i,j}(T)}|x^{i,j}(t) = x,\Sigma(t) = \Sigma  ],
\end{align}
we simply have $\phi^{i,j}(\omega,t, \tau, x ,\mathbf{V})=G^{i,j}(\mathtt{i}\omega,t, \tau, x ,\Sigma)$.
As a consequence, in our affine model it is sufficient to derive the discounted Laplace transform, which is explicitly given by the following proposition.

\begin{proposition}\label{LT_FW}Assume that the matrices $A_i,A_j$ are symmetric. Then, in the Wishart model, the discounted Laplace Transform of $x^{i,j}(t):=\log S^{i,j}(t)$ is given by:
\begin{align}
G^{i,j}\left(\omega,t,T,x,\Sigma\right)=exp\left[\omega x+\mathcal{A}(\tau)+Tr\left[\mathcal{B}(\tau)\Sigma\right]\right],\label{LT0_W}
\end{align}
where:
\begin{align}
\mathcal{A}&=\left(\omega\left(h^i-h^j\right)-h^i\right)\tau-\frac{\beta}{2}Tr\left[\log\mathcal{F}(\tau)+\left(\tilde{M}^{\top}+\omega\left(A_i-A_j\right)R^{\top}Q\right)\tau\right],\label{LT1_W}\\
\mathcal{B}(\tau)&=\mathcal{B}_{22}(\tau)^{-1}\mathcal{B}_{21}(\tau)\label{LT2_W}
\end{align}
and $\mathcal{B}_{22}(\tau),\mathcal{B}_{21}(\tau)$ are submatrices in:
\begin{align}
&\left(\begin{array}{cc}
	\mathcal{B}_{11}(\tau)&\mathcal{B}_{12}(\tau)\\
	\mathcal{B}_{21}(\tau)&\mathcal{B}_{22}(\tau)
\end{array}\right)=\nonumber\\
&\exp\tau\left[\begin{array}{cc}
	\tilde{M}+\omega Q^\top R\left(A_i-A_j\right)&-2Q^{\top}Q\\
	\frac{\omega^2-\omega}{2}\left(A_i-A_j\right)^2+(\omega-1)H^i-\omega H^j&-\left(\tilde{M}^{\top}+\omega\left(A_i-A_j\right)R^{\top}Q\right)\label{LT3_W}
\end{array}
\right].
\end{align}

\end{proposition}
\begin{proof} See the Appendix.\end{proof}

As we stated above, for $\omega=\mathtt{i}\lambda$ (where $\mathtt{i}=\sqrt{-1}$), we obtain the discounted characteristic function of the log-exchange rate, hence we can compute call option prices as e.g. in \cite{article_Carr99} via:
\begin{equation}
\frac{exp\left\{-\alpha c\right\}}{2\pi}\int_{-\infty}^{+\infty}\Re\left\{e^{-\mathtt{i}vc}\varphi(v)\right\}dv,
\end{equation}
where:
\begin{align}
\varphi(v)&=\frac{G\left(\mathtt{i}\left( v-\left(\alpha +1 \right)\mathtt{i}\right),t,T,x,V\right)}{\left(\alpha +\mathtt{i}v\right)\left(1+\alpha +\mathtt{i}v\right)}.
\end{align}

\subsection{Expansions}
The calibration of our model relies on a standard non-linear least squares procedure. This will be employed to minimize the distance between model and market implied volatilities. The model implied volatilities are extracted from the prices produced by the FFT routine. This procedure for the Wishart model is more demanding from a numerical point of view than the analog for the multi-Heston case e.g. in \cite{gnoatto11}. An alternative approach is to fit implied volatilities via a simpler function. A possibility is to find a relationship between the prices produced by the model, and the standard Black-Scholes formula. The next result states that it is possible to approximate the prices of options under the Wishart model, via a suitable expansion of the standard Black-Scholes formula and its derivatives, analogously to what has been done in \cite{gnoatto11}. The proof, that is reported in the appendix, relies on arguments that may be found in \cite{lewis2000} and \cite{dafgra11} (we drop all currency indices, it is intended that we are considering the $(i, j)$ FX pair). Define $\tau:=T-t$ and let us define the real deterministic functions $\tilde{\mathcal{B}}^0,\tilde{\mathcal{B}}^1,\tilde{\mathcal{B}}^{20},\tilde{\mathcal{B}}^{21}$ as follows:
\begin{align}
\tilde{\mathcal{B}}^0=&\int_{0}^{\tau}{e^{\left(\tau-u\right)\tilde{M}^{\top}}\left(A_i-A_j\right)e^{\left(\tau-u\right)\tilde{M}}du}\label{B0},\\
\tilde{\mathcal{B}}^1=&\int_{0}^{\tau}e^{\left(\tau-u\right)\tilde{M}^{\top}}\left(\tilde{\mathcal{B}}^0(u)Q^{\top}R\left(A_i-A_j\right)\right.\nonumber\\
&\left.+\left(A_i-A_j\right)R^{\top}Q\tilde{\mathcal{B}}^0(u)\right)e^{\left(\tau-u\right)\tilde{M}}du,\\
\tilde{\mathcal{B}}^{20}=&\int_{0}^{\tau}e^{\left(\tau-u\right)\tilde{M}^{\top}}2\tilde{\mathcal{B}}^0(u)Q^{\top}Q\tilde{\mathcal{B}}^0(u)e^{\left(\tau-u\right)\tilde{M}}du,\\
\tilde{\mathcal{B}}^{21}=&\int_{0}^{\tau}e^{\left(\tau-u\right)\tilde{M}^{\top}}\left(\tilde{\mathcal{B}}^1(u)Q^{\top}R\left(A_i-A_j\right)\right.\nonumber\\
&\left.+\left(A_i-A_j\right)R^{\top}Q\tilde{\mathcal{B}}^1(u)\right)e^{\left(\tau-u\right)\tilde{M}}du.\label{B21}
\end{align}
Moreover, the real deterministic scalar functions $\tilde{\mathcal{A}}^{0}(\tau),\tilde{\mathcal{A}}^{1}(\tau),\tilde{\mathcal{A}}^{20}(\tau),\tilde{\mathcal{A}}^{21}(\tau)$ are given by:
\begin{align}
\tilde{\mathcal{A}}^{0}(\tau)&=Tr\left[\Omega\Omega^{\top}\int_{0}^{\tau}{\tilde{\mathcal{B}}^0(u)du}\right],\label{mathcal_A0}\\
\tilde{\mathcal{A}}^{1}(\tau)&=Tr\left[\Omega\Omega^{\top}\int_{0}^{\tau}{\tilde{\mathcal{B}}^1(u)du}\right],\\
\tilde{\mathcal{A}}^{20}(\tau)&=Tr\left[\Omega\Omega^{\top}\int_{0}^{\tau}{\tilde{\mathcal{B}}^{20}(u)du}\right],\\
\tilde{\mathcal{A}}^{21}(\tau)&=Tr\left[\Omega\Omega^{\top}\int_{0}^{\tau}{\tilde{\mathcal{B}}^{21}(u)du}\right]\label{mathcal_A21}.
\end{align}
Finally, let
\begin{align} v=\sigma^2\tau=\tilde{\mathcal{A}}^0(\tau)+Tr\left[\tilde{\mathcal{B}}^0(\tau)\Sigma\right]\label{int_var_2}
\end{align}
be the integrated variance.

\begin{proposition}\label{exp_Wis1}
Assume that all interest rates are constant and the vol of vol matrix $Q$ has been scaled by the factor $\alpha>0$. Then the call price $C(S(t),K,\tau)$ in the Wishart-based exchange model can be approximated in terms of the vol of vol scale factor $\alpha$  by differentiating the Black Scholes formula $C_{B\&S}\left(S(t),K,\sigma, \tau\right)$ with respect to the log exchange rate $x(t)=\ln S(t)$ and the integrated variance $v=\sigma^2\tau$:
\begin{align}
C(S(t),K,\tau)
\approx &C_{B\&S}\left(S(t),K,\sigma,\tau\right)\nonumber\\
&+\alpha\left(\tilde{\mathcal{A}}^1(\tau)+Tr\left[\tilde{\mathcal{B}}^1(\tau)\Sigma(t)\right]\right)\partial^2_{xv}C_{B\&S}\left(S(t),K,\sigma,\tau\right)\nonumber\\
&+\alpha^2\left(\tilde{\mathcal{A}}^{20}(\tau)+Tr\left[\tilde{\mathcal{B}}^{20}(\tau)\Sigma(t)\right]\right)\partial^2_{v^2}C_{B\&S}\left(S(t),K,\sigma,\tau\right)\nonumber\\
&+\alpha^2\left(\tilde{\mathcal{A}}^{21}(\tau)+Tr\left[\tilde{\mathcal{B}}^{21}(\tau)\Sigma(t)\right]\right)\partial^3_{x^2v}C_{B\&S}\left(S(t),K,\sigma,\tau\right)\nonumber\\
&+\frac{\alpha^2}{2}\left(\tilde{\mathcal{A}}^{1}(\tau)+Tr\left[\tilde{\mathcal{B}}^{1}(\tau)\Sigma(t)\right]\right)^2\partial^4_{x^2v^2}C_{B\&S}\left(S(t),K,\sigma,\tau\right),
\end{align} 
\end{proposition}
\begin{proof} See the Appendix.\end{proof}

Finally,  we can state another formula, that does not involve the computation of option prices and constitutes an approximation of the implied volatility surface for a short time to maturity. This formula may be a useful alternative in order to get a quicker calibration for short maturities. 
\begin{proposition}\label{exp_Wis2}
Assume that all interest rates are constant. For a short time to maturity the implied volatility expansion in terms of the vol-of-vol scale factor $\alpha$ in the Wishart-based exchange model is given by:
\begin{align}
\sigma_{imp}^2\approx & Tr\left[\left(A_i-A_j\right)\Sigma(t)\left(A_i-A_j\right)\right]+\alpha\frac{Tr\left[\left(A_i-A_j\right)^2Q^{\top}R\left(A_i-A_j\right)\Sigma(t)\right]m_f}{Tr\left[\left(A_i-A_j\right)\Sigma(t)\left(A_i-A_j\right)\right]}\nonumber\\
&+\alpha^2\frac{m_f^2}{Tr\left[\left(A_i-A_j\right)\Sigma(t)\left(A_i-A_j\right)\right]^2}\Bigg[\frac{1}{3}Tr\left[\left(A_i-A_j\right)Q^\top Q\left(A_i-A_j\right)\Sigma(t)\right]\Bigg.\nonumber\\
&+\frac{1}{3}Tr\Big[\left[\left(A_i-A_j\right)Q^{\top}R\left(A_i-A_j\right)+\left(A_i-A_j\right)R^{\top}Q\left(A_i-A_j\right)\right]\Big.\nonumber\\
&\Bigg.\Big.\times Q^\top R\left(A_i-A_j\right)\Sigma(t)\Big]-\frac{5}{4}\frac{Tr\left[\left(A_i-A_j\right)Q^{\top}R\left(A_i-A_j\right)\Sigma(t)\right]^2}{Tr\left[\left(A_i-A_j\right)\Sigma(t)\left(A_i-A_j\right)\right]}\Bigg].
\end{align}
where $m_f=\log\left(\frac{S^{i,j}(t)e^{(r_i-r_j)\tau}}{K}\right)$ denotes the log-moneyness.
\end{proposition}
\begin{proof} See the Appendix. \end{proof}

\subsection{The pricing of zero-coupon Bonds}

The flexibility of the Wishart hybrid model opens up the possibility to perform a simultaneous calibration of interest rate and foreign exchange related products. This is a preliminary step that guarantees a coherent framework for the evaluation of payoffs depending on several interest rate curves under different currencies. Examples of such products may be found e.g. in \cite{brigo06} and \cite{bookClark11}. We provide a closed-form formula for the price of a zero-coupon bond, that constitutes a building block for many other linear interest rate products. The following proposition may be easily proved along the same lines of Proposition \ref{LT_FW}.

\begin{proposition}
The price of a zero-coupon bond at time $t$, with maturity $T$, under a generic risk neutral measure $\mathbb{Q}_d$, is given by
\begin{align}\label{bondPrice}
\mathbb{E}^{\mathbb{Q}_d}\left[\left.e^{-\int_t^T\left(h^d+Tr\left[H^d\Sigma(s)\right]\right)ds}\right|\mathcal{F}_t\right]=\exp\left\{\mathcal{A}^{ZC}(\tau)+Tr\left[\mathcal{B}^{ZC}(\tau)\Sigma(t)\right]\right\}
\end{align}
where the deterministic functions $\mathcal{A}^{ZC},\mathcal{B}^{ZC}$ satisfy the system of matrix ODE
\begin{align}
\frac{\partial \mathcal{A}^{ZC}}{\partial \tau}=&Tr\left[\beta Q^\top Q\mathcal{B}^{ZC}(\tau)\right]-h^d, \ \mathcal{A}^{ZC}(0)=0,\\
\frac{\partial \mathcal{B}^{ZC}}{\partial \tau}=&\mathcal{B}^{ZC}(\tau)M^{\mathbb{Q}_d}+M^{\mathbb{Q}_d,\top}\mathcal{B}^{ZC}(\tau)\nonumber\\
&+2\mathcal{B}^{ZC}(\tau)Q^\top Q\mathcal{B}^{ZC}(\tau)-H^d, \ \mathcal{B}^{ZC}(0)=0,
\end{align}
whose solution is given by
\begin{align}
\mathcal{A}^{ZC}(\tau)&=-\frac{\beta}{2}Tr\left[\log B^{ZC}_{22}(\tau)+\tau M^{\mathbb{Q}_d}\right]-h^d\tau,\\
\mathcal{B}^{ZC}(\tau)&=\mathcal{B}^{ZC}_{22}(\tau)^{-1}\mathcal{B}^{ZC}_{21}(\tau)\label{LT2_ZC}
\end{align}
and $\mathcal{B}^{ZC}_{22}(\tau),\mathcal{B}^{ZC}_{21}(\tau)$ are submatrices in:
\begin{align}
\left(\begin{array}{cc}
	\mathcal{B}^{ZC}_{11}(\tau)&\mathcal{B}^{ZC}_{12}(\tau)\\
	\mathcal{B}^{ZC}_{21}(\tau)&\mathcal{B}^{ZC}_{22}(\tau)
\end{array}\right)
=\exp\left[ \tau\left(\begin{array}{cc}
	M^{\mathbb{Q}_d}&-2Q^{\top}Q\\
	-H^d&-M^{\mathbb{Q}_d,\top}
\end{array}\right)
\right].\label{LT3_ZC}
\end{align}
\end{proposition}

Given this simple formula for zero-coupon bond, we may consider a joint calibration to the FX smile and to the risk-free curve of two different economies, so as to capture simultaneously the information coming from the FX smile and the interest rate curve, that plays an important role for long maturities. More precisely, a direct inspection of formula \eqref{bondPrice}, reveals that the Wishart short rate model belongs to the class of affine term structure models, meaning that the yield curve is of a particularly simple form. For fixed $\Sigma(t)$ we define the yield curve as the function
$Y:\mathbb{R}_{\geq 0}\rightarrow \mathbb{R}_{\geq 0}
$ with
\begin{align}
Y(\tau)=-\frac{1}{\tau}\left(\mathcal{A}^{ZC}(\tau)+Tr\left[\mathcal{B}^{ZC}(\tau)\Sigma(t)\right]\right)\label{ycdef}.
\end{align}
In Section \ref{caliFXIR} we provide a concrete calibration example involving the simultaneous fit of the FX surface and two yield-curves.

\subsection{The pricing of a Cap}
For the sake of completeness, we also report the pricing of an interest rate cap. We would like to point out that the formula presented here is derived under the assumption of a single curve interest-rate framework. A generalization involving the recent developments in the context of interest rate modelling, like multiple-curve models/OIS discounting, is beyond the scope of the paper.\\

Suppose we have a sequence of resetting dates $\left\{T_\gamma,\cdots,T_{\beta-1}\right\}$ and payment dates $\left\{T_{\gamma+1},\cdots,T_{\beta}\right\}$ at which the Libor rate $L(T_i,T_i,T_{i+1})=L(T_i,T_{i+1})$ observed at time $T_{i}$ for the time span between $T_{i}$ and $T_{i+1}$ is paid. The price at time $t<T_{\gamma}$ of an interest rate cap with notional $N$, strike price $K$ is given by

\begin{align}
C(t,K)=\sum_{i=\gamma+1}^{\beta}N\tau_{i}\mathbb{E}^{\mathbb{Q}_{d}}\left[\left.e^{-\int_{t}^{T_{i}}r_sds}\left(L(T_{i-1},T_{i})-K\right)^+\right|\mathcal{F}_t\right].
\end{align}
It is well known that the expectation above may be conveniently rewritten as a put option on a zero coupon bond
\begin{align}
Cap(t,K)&=\sum_{i=\gamma+1}^{\beta}N\left(1+\tau_iK\right)\nonumber\\
&\times\mathbb{E}^{\mathbb{Q}_{d}}\left[\left.e^{-\int_{t}^{T_{i-1}}r_sds}\left(\frac{1}{1+\tau_iK}-P(T_{i-1},T_{i})\right)^+\right|\mathcal{F}_t\right].
\end{align}

Given the nice analytical properties of the Wishart process, we would like to compute this expectation using Fourier methods. The standard approach, in the presence of such an expectation, would involve a change to a forward risk-neutral measure. While this is also possible in this case, we remark that the resulting dynamics of the Wishart process would involve time varying coefficients, which in turn would lead to the solution of time dependent  Matrix Riccati differential equations, a highly non trivial problem. Fortunately, this kind of problem can be conveniently overcome, since, in our setting, it is easy to compute expectations that involve the discount factor and the zero-coupon bond. In fact, without loss of generality, let $\tau_i=\tau, \ \forall i$ and introduce a positive constant $\alpha > 0$ and $k=\log(\frac{1}{1+\tau K})$: from e.g. \cite{article_Carr99} it follows that the price of a put on a zer-coupon bond can be written as

\begin{align}
Put^\alpha(t,k)=\frac{e^{\alpha k}}{2\pi}\int_0^{\infty}\Re\left\{e^{-\mathtt{i}vk}\psi(v)\right\}dv
\end{align}
where
\begin{align}
\psi(v)=\frac{\mathbb{E}^{\mathbb{Q}_d}\left[\left.e^{-\int_{t}^{T_{i-1}}r_sds}e^{\mathtt{i}\left(v-(-\alpha+1)\mathtt{i}\right)\log P(T_{i-1},T_{i})}\right|\mathcal{F}_t\right]}{\alpha^2-\alpha-v^2+\mathtt{i}\left(-2\alpha+1\right)v}.
\end{align}

In the denominator of the expression above, we recognize the time $t$ conditional discounted characteristic function of the zero coupon bond at time $T_{i-1}$ having maturity $T_{i}$ evaluated at the point $v-(-\alpha+1)\mathtt{i}$. While this object may seem difficult to compute, we remarkt that under our assumption on the shape of the short interest rate
\begin{align}
r_t=h^d+Tr\left[H^d\Sigma(t)\right],
\end{align}
we can perform the computation, by relying on the techniques we have been already employing.

\begin{proposition}
The time $t$ conditional discounted characteristic function of the logarithmic zero-coupon bond price with maturity $T_{i}$, at time $T_{i-1}$, evaluated at the point $\omega$ is given by

\begin{align}
&\mathbb{E}^{\mathbb{Q}_d}\left[\left.e^{-\int_t^{T_{i-1}}\left(h^d+Tr\left[H^d\Sigma(S)\right]\right)ds+\omega\log P(T_{i-1},T_{i})}\right|\mathcal{F}_t\right]\nonumber\\
&=\exp\left\{\omega\mathcal{A}^{ZC}(T_{i}-T_{i-1})+\mathcal{A}^{LZC}(T_{i-1}-t)\right.\nonumber\\
&\left.+Tr\left[\mathcal{A}^{LZC}(T_{i-1}-t)\Sigma(t)\right]\right\}
\end{align}
where the deterministic functions $\mathcal{A}^{LZC},\mathcal{B}^{LZC}$ satisfy the system of matrix ODE
\begin{align}
\frac{\partial \mathcal{A}^{LZC}}{\partial \tau}&=Tr\left[\beta Q^\top Q\mathcal{B}^{LZC}(\tau)\right]-h^d, \ \mathcal{A}^{LZC}(0)=0,\\
\frac{\partial \mathcal{B}^{LZC}}{\partial \tau}&=\mathcal{B}^{LZC}(\tau)M^{\mathbb{Q}_d}+M^{\mathbb{Q}_d,\top}\mathcal{B}^{LZC}(\tau)\nonumber\\
&+2\mathcal{B}^{LZC}(\tau)Q^\top Q\mathcal{B}^{LZC}(\tau)-H^d, \ \mathcal{B}^{LZC}(0)=\omega\mathcal{B}^{ZC}(T_{i}-T_{i-1}),
\end{align}
whose solution is give by
\begin{align}
\mathcal{A}^{LZC}(\tau)=-\frac{\beta}{2}Tr\left[\log B^{LZC}_{22}(\tau)+\tau M^{\mathbb{Q}_d}\right]-h^d\tau,
\end{align}
\begin{align}
\mathcal{B}^{LZC}(\tau)&=\left(\omega\mathcal{B}^{ZC}(T_{i}-T_{i-1})\mathcal{B}^{LZC}_{12}(\tau)+\mathcal{B}^{LZC}_{22}(\tau)\right)^{-1}\nonumber\\
&\times\left(\omega\mathcal{B}^{ZC}(T_{i}-T_{i-1})\mathcal{B}^{LZC}_{11}(\tau)+\mathcal{B}^{LZC}_{21}(\tau)\right)\label{LT2_LZC}
\end{align}
and $\mathcal{B}^{LZC}_{22}(\tau),\mathcal{B}^{LZC}_{21}(\tau)$ are submatrices in:
\begin{align}
\left(\begin{array}{cc}
	\mathcal{B}^{LZC}_{11}(\tau)&\mathcal{B}^{LZC}_{12}(\tau)\\
	\mathcal{B}^{LZC}_{21}(\tau)&\mathcal{B}^{LZC}_{22}(\tau)
\end{array}\right)
=\exp\left[ \tau\left(\begin{array}{cc}
	M^{\mathbb{Q}_d}&-2Q^{\top}Q\\
	-H^d&-M^{\mathbb{Q}_d,\top}
\end{array}\right)
\right].\label{LT3_LZC}
\end{align}

\end{proposition}

With the previous result, caps and floors can be easily priced within our framework.

\section{The multi-Heston model of \cite{gnoatto11} as a nested approach}

In this section we follow \cite{ben2010} in order to show that when all the matrices $M,Q,R,\Sigma_(0),A_i,H^i$ are \textbf{diagonal} and for a particular specification of the noise $Z$ driving the exchange rates, the Wishart model becomes equivalent to the multi-Heston approach introduced in \cite{gnoatto11}. In this particular case the dynamics of the elements of the Wishart process take the simpler form
\begin{align*}
d\Sigma_{pp}(t)=\left(\beta Q^2_{pp}+2M_{ii}\Sigma_{pp}(t)\right)dt+2Q_{pp}\sum_{k=1}^{d}\sqrt{\Sigma(t)}_{pk}dW^{\mathbb{Q}^{i}}_{kp}(t)
\end{align*}
and
\begin{align*}
d\Sigma_{pm}(t)=&\Sigma_{pm}(t)\left(M_{pp}+M_{mm}\right)dt+Q_{pp}\sum_{k=1}^{d}\sqrt{\Sigma(t)}_{km}dW^{\mathbb{Q}^{i}}_{kp}(t)\\
&+Q_{mm}\sum_{k=1}^{d}\sqrt{\Sigma(t)}_{pk}dW^{\mathbb{Q}^{i}}_{km}(t).
\end{align*}
We notice that
\begin{align*}
d\left\langle \Sigma_{pp}(t),\Sigma_{pp}(t) \right\rangle=4Q_{pp}\Sigma_{pp}(t)dt,
\end{align*}
so the $d$-dimensional process $\tilde{W}=\left(\tilde{W}(t)\right)_{t\geq 0}$, with
\begin{align*}
d\tilde{W}^{\mathbb{Q}^i}_p(t)=\sqrt{\Sigma_{pp}}\sum_{k=1}^{d}\sqrt{\Sigma(t)}_{pk}dW^{\mathbb{Q}^i}_{kp(t)}
\end{align*}
is a vector of independent Brownian motions. As a consequence, we may conveniently express the dynamics of a generic diagonal element as
\begin{align*}
d\Sigma_{pp}(t)=\left(\beta Q^2_{pp}+2M_{ii}\Sigma_{pp}(t)\right)dt+2Q_{pp}\sqrt{\Sigma_{pp}(t)}d\tilde{W}^{\mathbb{Q}^i}_p(t).
\end{align*}

Assume also that the matrix  $Z^{\mathbb{Q}^i}=\left(Z^{\mathbb{Q}^{i}}(t)\right)_{t\geq 0}$ is diagonal\footnote{Note that $Z$ is no more a matrix Brownian motion, which by definition consists in $d\times d$ independent Brownian motions.}.  A quick inspection at the diffusion part of the dynamics of the exchange rate $S^{i,j}$ reveals that we are considering a setting that is equivalent to the multi-Heston model of \cite{gnoatto11}, namely, we may conveniently rewrite the dynamics 
\begin{align*}
\frac{dS^{i,j}(t)}{S^{i,j}(t)}&=(r^{i}(t)-r^{j}(t))dt+Tr\left[(A_i -A_j)\sqrt{\Sigma(t)}dZ^{\mathbb{Q}^i}(t)\right]
\end{align*}
as
\begin{align*}
\frac{dS^{i,j}(t)}{S^{i,j}(t)}&=(r^{i}(t)-r^{j}(t))dt+({\bf a}_i -{\bf a}_j)^{\top}\sqrt{Diag\left(\Sigma(t)\right)}d\tilde Z^{\mathbb{Q}^i}(t),
\end{align*}
where ${\bf a}_i,{\bf a}_j$ are column vectors, $\tilde Z^{\mathbb{Q}^i}(t)$ is a $d$-dimensional vector Brownian motion with $\tilde Z_p^{\mathbb{Q}^i}(t)= Z_{pp}^{\mathbb{Q}^i}(t)$ and $\sqrt{Diag\left(\Sigma(t)\right)}$ denotes a diagonal matrix, featuring the elements $\sqrt{\Sigma_{pp}(t)}, \ p=1,...,d$ along the main diagonal. The short rate processes may be equivalently expressed by means of a scalar product between a vector featuring the elements of the main diagonal of $H^i$ and a second vector featuring the elements of the main diagonal of $\Sigma$. This allows us to connect the finite dimensional distributions generated by the hybrid Wishart process with those generated by a multi-Heston model. Being the finite dimensional distributions equal, also the prices of European options will be the same under the two specifications. From this informal discussion, we also deduce that the Wishart hybrid model nests an extension of the approach of \cite{gnoatto11} featuring stochastic interest rates.

\section{Numerical illustration}\label{calibration}

\subsection{Joint calibration of the FX-IR hybrid model}\label{caliFXIR}

In this subsection we report the results of a simultaneous calibration of a foreign exchange options volatility surface and of the two yield curves of the economies linked by the foreign exchange rate. To be more specific, we consider an implied volatility surface of $EUR/USD$ and the yield curves of $EUR$ and $USD$ economies.

\subsubsection{Description of the data}

We consider market data on January 15$^{th}$ 2013. The data set features an implied volatility surface for options written on $EUR/USD$. Moneyness ranges from $5\Delta$ put up to $5\Delta$ call. Maturities range from 1 day till 15 years. The conversion between deltas and strikes can be easily performed along the lines of \cite{bookClark11} or \cite{Wystup10}. As far as the interest rate market is concerned, we obtained the yield curves for both $EUR$ and $USD$. In this case we considered maturities up to 20 years for both curves.

\subsubsection{Calibration results}

The idea of the present calibration is to try to fit simultaneously the following data: the market FX implied volatility surface of $EUR/USD$ and the $EUR$ and $USD$ yield curves. Our calibration is thought of as instrumental for the evaluation of long-dated FX products like power reverse dual currency notes (PRDC, see the introduction). For this reason, the range of expiries that we consider is much larger. As far as the yield curves are concerned, we fit market data up to 20 years. Concerning the implied volatility surface, we consider maturities ranging from 1 month up to 15 years (the longest maturity at which traders are allowed to trade). The penalty function is constructed first by considering the distance between market and model implied volatilities for FX options and then by also looking at the distance between market yields and model yields computed according to formula \eqref{ycdef}.\\

Figures \ref{fig:stoch_rate_1} and \ref{fig:stoch_rate_2} report the results of our calibration. The fit is very satisfactory across different strikes and maturities for FX options. The upper part or Figure \ref{fig:stoch_rate_1} shows the actual distance between the two surfaces, that are almost overlapping. For the sake of readability in the bottom left part of Figure \ref{fig:stoch_rate_1} we multiply by a factor $>1$ all model implied volatilities so as to ease the comparison. The bottom right part of Figure \ref{fig:stoch_rate_1} provides a more precise view on the quality of the calibration of the FX surface, by plotting for each point in the maturity/delta space the squared difference between model and market implied volatilities. Figure \ref{fig:stoch_rate_2} reports a comparison between market and model implied yield curves. Red stars and blue circles denote model and market implied yields respectively. Even in this market we are able to obtain a satisfactory fit. Recall that we perform a joint calibration procedure, meaning that the same set of model parameters allows to obtain the results of Figures \ref{fig:stoch_rate_1} and \ref{fig:stoch_rate_2}.\\

 The values of the parameters arising from the calibration are reported in Table \ref{tab:paramsWisIRFX}.
By observing the values of the parameters we obtained, we notice that for the parameter $\beta$ we have $\beta\geq d+1$, hence for the fitted model we conclude that the dollar and the euro risk-neutral measures are well posed martingale measures. As far as the values of $h_{EUR}$ and $h_{USD}$ are concerned, they are negative, which is not in line with our starting assumption. In a first calibration experiment, we imposed the constraint $h_i>0, \ i=USD,EUR$ but observed that the model was not able to replicate the observed yield curve shape. 
Recall that in the present framework we are using (projections of) the Wishart process for the fit on both the FX implied volatility surface and the yield curves. This fact results in a trade-off between the parameters of the model. In other words, a small value of the initial state variables $\Sigma$ conflicts with the ability of the model to fit the short term smile. This is in line with the findings of \cite{article_Chiarella} who also found negative values for $h$, thus implying a distribution of the short rate that can become negative. Anyway, upon relaxation of the positivity constraint we obtained the very satisfactory fit that we reported in the present section.

\section{Conclusions}\label{conclusions}
In this paper we introduced a novel hybrid model that allows a joint evaluation of interest rate products and FX derivatives. The model is based on the Wishart process and so it retains the same level of analytical tractability typical of the affine class, like in the multi-Heston version introduced in \cite{gnoatto11}. In analogy with \cite{gnoatto11}, we have a model that is theoretically consistent with the triangular relationship among FX rates and other stylized symmetries that are commonly observed in the FX market. Moreover, the model allows for interest rate risk since there is the presence of stochastic interest rates that are non trivially correlated with the volatility structure of the exchange rates. A successful calibration of the $EUR/USD$ implied volatility surface and of the $EUR$ and $USD$ yield-curves suggests that the framework we propose is a suitable instrument for the evaluation of long-dated FX products, where a joint description of interest rates and FX markets is required.\\

There is ample room for future research. For example, the introduction of $S_d^+$ valued jump processes, as in e.g. \cite{article_MPS02}, would help in capturing the highly skewed implied volatilities that are usually observed for short maturities. Another interesting direction is the generalization to the multi-curve setting that is emerging in interest rate modelling as a response to the recent financial crisis. Several authors have attempted a solution in the FX context in order to model different levels of risks between e.g. EONIA and EURIBOR curves (see e.g. \cite{bia10}, \cite{fries10}, \cite{ken10}, \cite{kitawo09} and references therein). We believe that a Wishart based approach would give very interesting insights in this multivariate puzzle.

\section{Proofs}\label{proofs2}
\subsection{Proof of Proposition \ref{prop_6}}
Application of the Ito formula to the product $S^{i,l}(t)S^{l,j}(t)$ (using the property $Tr[dWA]Tr[dWB]=Tr[AB]$) leads to (\ref{Wrates}). In formulas:
\begin{align*}
\frac{dS^{i,l}(t)}{S^{i,l}(t)}&=(r^i-r^l)dt+Tr[(A_i -A_l)\Sigma(t)A_i]dt+Tr[(A_i -A_l)\sqrt{\Sigma(t)}dZ(t)],\nonumber\\
\frac{dS^{l,j}(t)}{S^{l,j}(t)}&=(r^l-r^j)dt+Tr[(A_l -A_j)\Sigma(t)A_l]dt+Tr[(A_l -A_j)\sqrt{\Sigma(t)}dZ(t)],\nonumber
\end{align*}
\begin{align*}
dS^{i,j}(t)=& dS^{i,l}(t)_{t}S^{l,j}_{t}+S^{i,l}(t)_{t}dS^{l,j}_{t}+d\left\langle S^{i,l},S^{l,j} \right\rangle_t\nonumber\\
=&S^{l,j}_{t}S^{i,l}(t)\Bigg((r^i-r^l)dt+Tr[(A_i -A_l)\Sigma(t)A_i]dt+Tr[(A_i -A_l)\sqrt{\Sigma(t)}dZ(t)]\Bigg)\nonumber\\
&+S^{l,j}_{t}S^{i,l}(t)\Bigg((r^l-r^j)dt+Tr[(A_l -A_j)\Sigma(t)A_l]dt+Tr[(A_l -A_j)\sqrt{\Sigma(t)}dZ(t)]\Bigg)\nonumber\\
&+d\left\langle \int_0^.Tr[(A_i -A_l)\sqrt{\Sigma(u)}dZ(u)],\int_0^.Tr[(A_l -A_j)\sqrt{\Sigma(u)}dZ(u)] \right\rangle_t.
\end{align*}
We concentrate on the covariation term:
\begin{align*}
&d\left\langle \int_0^.Tr[(A_i -A_l)\sqrt{\Sigma(v)}dZ(v)],\int_0^.Tr[(A_l -A_j)\sqrt{\Sigma(v)}dZ(v)] \right\rangle_t\nonumber\\
=&d\left\langle \int_0^.\sum_{p,q,r=1}^{d}{(A_i -A_l)_{pq}\sqrt{\Sigma(v)}_{qr}dZ(v)_{{rp}}},\int_0^.\sum_{s,t,u=1}^{d}{(A_l -A_j)_{st}\sqrt{\Sigma(v)}_{tu}dZ(v)_{{us}}}\right\rangle_t\nonumber\\
=&\sum_{p,q,r,s,t,u=1}^{d}(A_i -A_l)_{pq}\sqrt{\Sigma(t)}_{qr}dZ(t)_{{rp}}\nonumber\\
&\times(A_l -A_j)_{st}\sqrt{\Sigma(t)}_{tu}dZ(t)_{{us}}\delta_{r=u}\delta_{p=s}\nonumber\\
=&\sum_{s,q,u,t=1}^{d}{(A_i -A_l)_{sq}\sqrt{\Sigma(t)}_{qu}(A_l -A_j)_{st}\sqrt{\Sigma(t)}_{tu}dt}\nonumber\\
=&\sum_{s,q,u,t=1}^{d}{(A_i -A_l)_{sq}\sqrt{\Sigma(t)}_{qu}\sqrt{\Sigma(t)}_{ut}(A_l^\top -A_j^\top)_{ts}dt}.
\end{align*}
By assuming that the matrices $Ai,A_l,A_j$ are symmetric we get
\begin{align*}
&d\left\langle \int_0^.Tr[(A_i -A_l)\sqrt{\Sigma(v)}dZ(v)],\int_0^.Tr[(A_l -A_j)\sqrt{\Sigma(v)}dZ(v)] \right\rangle_t\nonumber\\
&=Tr\left[(A_i -A_l)\Sigma(t)(A_l -A_j)\right].
\end{align*}
Finally, using the fact that terms in the trace commute we obtain:
\begin{align*}
\frac{dS^{i,j}(t)}{S^{i,j}(t)}&=(r^i-r^j)dt+Tr[(A_i -A_j)\Sigma(t)A_i]dt\nonumber\\
&+Tr[(A_i -A_j)\sqrt{\Sigma(t)}dZ(t)]\nonumber\\
\end{align*}
as desired.

\subsection{Proof of Proposition \ref{LT_FW}}
Recall that $\phi^{i,j}\left(\omega,t,\tau,x,\Sigma\right)=G^{i,j}(i\omega,t,\tau,x,\Sigma)$ where $G$ was previously defined as $G^{i,j}\left(\omega,t,\tau,x,\Sigma\right)=\mathbb{E}_{t}^{\mathbb{Q}^i}\left[e^{\omega x_T}\right]$, $x^{i,j}(t):=\log{S^{i,j}(t)}$. These functions represent the conditional discounted characteristic and  moment generating functions of the log-exchange rate respectively. We follow closely \cite{article_DaFonseca1} in order to determine these quantities, so we first write the PDE satisfied by $G$, which requires the dynamics of $x(t)=x^{i,j}(t)$ under the measure $\mathbb{Q}^i$:
\begin{align}
d\log{S^{i,j}(t)}&=\left(\left(r^i-r^j\right)-\frac{1}{2}Tr\left[\left(A_i-A_j\right)\Sigma(t)\left(A_i-A_j\right)\right]\right)dt\nonumber\\
&+Tr\left[\left(A_i-A_j\right)\sqrt{\Sigma(t)}dZ^{\mathbb{Q}^i}(t)\right],
\end{align}
where the short rates are driven by the Wishart process in line with \eqref{shortRates}. The Laplace transform solves the following PDE in terms of $\tau=T-t$:
\begin{align}
\frac{\partial}{\partial \tau}G^{i,j}&=\mathcal{A}_{x,\Sigma}G^{i,j}-r^iG^{i,j},\label{PDEWi}\\
G^{i,j}(\omega,T,0,x,\Sigma)&=e^{\omega x}.
\end{align}
To solve this PDE we first determine the infinitesimal generator $\mathcal{A}_{x,\Sigma}$. This will feature the contribution of three terms: the process $x$, the Wishart process $\Sigma$ and the mixed term that corresponds to the coefficient of the term $\frac{\partial^2}{\partial x\partial \Sigma_{pt}}$ and arises from the correlation structure. The first is trivial, the second is known thanks to \cite{article_Bru} and is of the form:
\begin{align*}
Tr\left[\left(\Omega\Omega^{\top}+\tilde{M}\Sigma+\Sigma \tilde{M}^{\top}\right)D+2\Sigma DQ^{\top}QD\right],
\end{align*}
where $D$ is the differential operator:
\begin{align*}
D_{pt}=\frac{\partial}{\partial \Sigma_{pt}}.
\end{align*}
In order to compute the mixed term, we notice that under the measure $\mathbb{Q}^i$:
\begin{align}
&d\left\langle \log{S^{i,j}},\Sigma_{pt}\right\rangle_u\nonumber\\
&=2d\left\langle\int_0^. Tr\left[\left(A_i-A_j\right)\sqrt{\Sigma(v)}dZ^{\mathbb{Q}^i}(v)\right],\int_0^.\sum_{q,r,s=1}^{d}{\sqrt{\Sigma(v)}_{pq}dZ_{qr}^{\mathbb{Q}^i}(v)R^{\top}_{rs}Q_{st}}\right\rangle_u\nonumber\\
&=2d\left\langle \int_0^.\sum_{a,b,c,=1}^{d}{\left(A_i-A_j\right)_{ab}\sqrt{\Sigma(v)}_{bc}dZ_{ca}^{\mathbb{Q}^i}(v)},\int_0^.\sum_{q,r,s=1}^{d}{\sqrt{\Sigma}_{pq}(v)dZ_{qr}^{\mathbb{Q}^i}(v)R^{\top}_{rs}Q_{st}}\right\rangle_u\delta_{c=q}\delta_{a=r}\nonumber\\
&=2\sum_{b,q,r,s=1}^{d}{\left(A_i-A_j\right)_{rb}\sqrt{\Sigma(t)}_{bq}\sqrt{\Sigma(t)}_{pq}R^{\top}_{rs}Q_{st}}dt\nonumber\\
&=2\sum_{b,q,r,s=1}^{d}{\Sigma(t)_{pb}\left(A_i-A_j\right)_{br}R^{\top}_{rs}Q_{st}}dt,\nonumber
\end{align}
where we have used the fact that
\begin{align*}
2Tr\left[\Sigma\left(A_i-A_j\right)R^{\top}QD\right]\frac{\partial}{\partial x}=2\sum_{p,b,r,s=1}^{d}{D_{tp}\Sigma_{pb}\left(A_i-A_j\right)_{br}R^{\top}_{rs}Q_{st}}\frac{\partial}{\partial x}
\end{align*}
and that $D$ is symmetric.
Now we can state the PDE satisfied by $G$:
\begin{align}
\frac{\partial G}{\partial \tau}=&\left(\left(r^i-r^j\right)-\frac{1}{2}Tr\left[\left(A_i-A_j\right)\Sigma\left(A_i-A_j\right)\right]\right)\frac{\partial G}{\partial x}\nonumber\\
&+\frac{1}{2}Tr\left[\left(A_i-A_j\right)\Sigma\left(A_i-A_j\right)\right]\frac{\partial^2 G}{\partial x^2}\nonumber\\
&+Tr\left[\left(\Omega\Omega^{\top}+\tilde{M}\Sigma+\Sigma \tilde{M}^{\top}\right)DG+2\left(\Sigma DQ^{\top}QD\right)G\right]\nonumber\\
&+2Tr\left[\Sigma\left(A_i-A_j\right)R^{\top}QD\right]\frac{\partial G}{\partial x}-h^i-Tr\left[H^i\Sigma\right].\label{PDE_SV}
\end{align}
Since the Wishart process $\Sigma$ is affine we guess the following
\begin{align}
G^{i,j}\left(\omega,t,\tau,x,\Sigma\right)=exp\left[\mathcal{C}(\tau)x+\mathcal{A}(\tau)+Tr\left[\mathcal{B}(\tau)\Sigma\right]\right],\label{guess2}
\end{align}
with $\mathcal{C},\mathcal{A}\in\mathbb{R}$ and $\mathcal{B}\in\mathcal{S}_d$ s.t. the transform is well defined, moreover these functions satisfy the following terminal conditions:
\begin{align*}
\mathcal{A}(0)&=0\in\mathbb{R},\\
\mathcal{C}(0)&=\omega\in\mathbb{R},\\
\mathcal{B}(0)&=0\in S_d.
\end{align*}
We substitute the candidate \eqref{guess2} into \eqref{PDE_SV} and obtain
\begin{align*}
\frac{\partial}{\partial \tau}\mathcal{C}(\tau)=0,
\end{align*}
hence: $\mathcal{C}(\tau)=\omega$ $\forall\tau$. We also have the following (matrix) Riccati ODE:
\begin{align}
\frac{\partial}{\partial \tau}\mathcal{B}&=\mathcal{B}(\tau)\left(\tilde{M}+\omega Q^\top R\left(A_i-A_j\right)\right)+\left(\tilde{M}^{\top}+\omega\left(A_i-A_j\right)R^{\top}Q\right)\mathcal{B}(\tau)\label{matRic}\nonumber\\
&+2\mathcal{B}(\tau)Q^{\top}Q\mathcal{B}(\tau)+\frac{\omega^2-\omega}{2}\left(A_i-A_j\right)^2+(\omega-1)H^i-\omega H^j
\end{align}
and the final ODE which may then be solved upon direct integration
\begin{align}
\frac{\partial}{\partial \tau}\mathcal{A}&=\omega\left(h^i-h^j\right)-h^i+Tr\left[\Omega\Omega^{\top}\mathcal{B}(\tau)\right].\label{ODE_for_A}
\end{align}
Following \cite{gra08}, it is possible to linearize \eqref{matRic} by writing
\begin{align}
\mathcal{B}(\tau)=\mathcal{F}^{-1}(\tau)\mathcal{G}(\tau),\label{linear}
\end{align}
for $\mathcal{F}(\tau)\in GL(d)$ and $\mathcal{G}(\tau)\in M_d$, then we have
\begin{align}
\frac{\partial}{\partial \tau}\mathcal{F}&=-\mathcal{F}(\tau)\left(\tilde{M}^{\top}+\omega\left(A_i-A_j\right)R^{\top}Q\right)-2\mathcal{G}(\tau)Q^{\top}Q\label{matF}\\
\frac{\partial}{\partial \tau}\mathcal{G}&=\mathcal{G}(\tau)\left(\tilde{M}+\omega Q^\top R\left(A_i-A_j\right)\right)\nonumber\\
&+\mathcal{F}(\tau)\left(\frac{\omega^2-\omega}{2}\left(A_i-A_j\right)^2+(\omega-1)H^i-\omega H^j\right),
\end{align}
with $\mathcal{F}(0)=I_d$ and $\mathcal{G}(0)=\mathcal{B}(0)$. The solution of the above system is
\begin{align*}
&\left(\mathcal{B}(0),I_d\right)
\left(\begin{array}{cc}
	\mathcal{B}_{11}(\tau)&\mathcal{B}_{12}(\tau)\\
	\mathcal{B}_{21}(\tau)&\mathcal{B}_{22}(\tau)
\end{array}\right)=\left(\mathcal{B}(0),I_d\right)\nonumber\\
&\times\exp \tau\left[\begin{array}{cc}
	\tilde{M}+\omega Q^\top R\left(A_i-A_j\right)&-2Q^{\top}Q\\
	\frac{\omega^2-\omega}{2}\left(A_i-A_j\right)^2+(\omega-1)H^i-\omega H^j&-\left(\tilde{M}^{\top}+\omega\left(A_i-A_j\right)R^{\top}Q\right)
\end{array}
\right]
\end{align*}
so that the solution for $\mathcal{B}(\tau)$ is
\begin{align*}
\mathcal{B}(\tau)=\left(\mathcal{B}(0)\mathcal{B}_{12}(\tau)+\mathcal{B}_{22}(\tau)\right)^{-1}\left(\mathcal{B}(0)\mathcal{B}_{11}(\tau)+\mathcal{B}_{21}(\tau)\right).
\end{align*}
Since $\mathcal{B}(0)=0$ we finally get $\mathcal{B}(\tau)=\mathcal{B}_{22}(\tau)^{-1}\mathcal{B}_{21}(\tau)$. Now starting  from \eqref{matF} we write
\begin{align*}
-\frac{1}{2}\left(\frac{\partial}{\partial \tau}\mathcal{F}+\mathcal{F}(\tau)\left(\tilde{M}^{\top}+\omega\left(A_i-A_j\right)R^{\top}Q\right)\right)\left(Q^{\top}Q\right)^{-1}&=\mathcal{G}(\tau).
\end{align*}
We plug this into \eqref{linear}, then we insert the resulting formula for $\mathcal{B}(\tau)$ into \eqref{ODE_for_A} and obtain
\begin{align}
\frac{\partial}{\partial \tau}\mathcal{A}&=\omega\left(h^i-h^j\right)-h^i\nonumber\\
&+Tr\left[-\frac{\beta}{2}\left(\mathcal{F}^{-1}(\tau)\frac{\partial}{\partial \tau}\mathcal{F}+\left(\tilde{M}^{\top}+\omega\left(A_i-A_j\right)R^{\top}Q\right)\right)\right]\nonumber
\end{align}
whose solution is
\begin{align*}
\mathcal{A}=\left(\omega\left(h^i-h^j\right)-h^i\right)\tau-\frac{\beta}{2}Tr\left[\log\mathcal{F}(\tau)+\left(\tilde{M}^{\top}+\omega\left(A_i-A_j\right)R^{\top}Q\right)\tau\right].
\end{align*}

\subsection{Proof of Proposition \ref{exp_Wis1}}
The matrix Riccati ODE \eqref{matRic} may be rewritten as follows after replacing $Q$ by a small perturbation $\alpha Q, \alpha\in\mathbb{R}$:
\begin{align}
\frac{\partial}{\partial \tau}\mathcal{B}=&\mathcal{B}(\tau)\left(\tilde{M}+\alpha\omega Q^{\top}R\left(A_i-A_j\right)\right)+\left(\tilde{M}^{\top}+\omega\alpha\left(A_i-A_j\right)R^{\top}Q\right)\mathcal{B}(\tau)\nonumber\\
&+2\alpha^2\mathcal{B}(\tau)Q^{\top}Q\mathcal{B}(\tau)+\frac{\omega^2-\omega}{2}\left(A_i-A_j\right)^2\\
\mathcal{B}(0)=&0.
\end{align}
We consider now an expansion in terms of $\alpha$ of the form $\mathcal{B}=\mathcal{B}^0+\alpha \mathcal{B}^1+\alpha^2 \mathcal{B}^2$. We substitute this expansion and identify terms by powers of $\alpha$. We obtain the following ODE's.
\begin{align}
\frac{\partial}{\partial \tau}\mathcal{B}^0=&\mathcal{B}^0(\tau)\tilde{M}+\tilde{M}^{\top}\mathcal{B}^0(\tau)+\frac{\omega^2-\omega}{2}\left(A_i-A_j\right)^2\\
\frac{\partial}{\partial \tau}\mathcal{B}^1=&\mathcal{B}^1(\tau)\tilde{M}+\tilde{M}^{\top}\mathcal{B}^1(\tau)+\mathcal{B}^0(\tau)Q^{\top}R\left(A_i-A_j\right)\omega\nonumber\\
&+\omega\left(A_i-A_j\right)R^{\top}Q\mathcal{B}^0(\tau)\\
\frac{\partial}{\partial \tau}\mathcal{B}^2=&\mathcal{B}^2(\tau)\tilde{M}+\tilde{M}^{\top}\mathcal{B}^2(\tau)+\mathcal{B}^1(\tau)Q^{\top}R\left(A_i-A_j\right)\omega\nonumber\\
&+\omega\left(A_i-A_j\right)R^{\top}Q\mathcal{B}^1(\tau)+2\mathcal{B}^0(\tau)Q^{\top}Q\mathcal{B}^0(\tau).
\end{align}
Let $\gamma:=\frac{\omega^2-\omega}{2}$ then these equations admit the following solutions
\begin{align}
\mathcal{B}^0(\tau)=&\frac{\omega^2-\omega}{2}\int_{0}^{\tau}{e^{\left(\tau-u\right)\tilde{M}^{\top}}\left(A_i-A_j\right)^2e^{\left(\tau-u\right)\tilde{M}}du}\nonumber\\
:=&\gamma \tilde{\mathcal{B}}^0(\tau),\\
\mathcal{B}^1(\tau)=&\gamma\omega\int_{0}^{\tau}e^{\left(\tau-u\right)\tilde{M}^{\top}}\left(\tilde{\mathcal{B}}^0(u)Q^{\top}R\left(A_i-A_j\right)\right.\nonumber\\
&\left.+\left(A_i-A_j\right)R^{\top}Q\tilde{\mathcal{B}}^0(u)\right)e^{\left(\tau-u\right)\tilde{M}}du\nonumber\\
:=&\gamma\omega\tilde{\mathcal{B}}^1(\tau),\\
\mathcal{B}^2(\tau)=&\gamma\omega^2\int_{0}^{\tau}e^{\left(\tau-u\right)\tilde{M}^{\top}}\left(\tilde{\mathcal{B}}^1(u)Q^{\top}R\left(A_i-A_j\right)\right.\nonumber\\
&\left.+\left(A_i-A_j\right)R^{\top}Q\tilde{\mathcal{B}}^1(u)\right)e^{\left(\tau-u\right)\tilde{M}}du\nonumber\\
&+\gamma^2\int_{0}^{\tau}e^{\left(\tau-u\right)\tilde{M}^{\top}}2\tilde{\mathcal{B}}^0(u)Q^{\top}Q\tilde{\mathcal{B}}^0(u)e^{\left(\tau-u\right)\tilde{M}}du\nonumber\\
:=&\gamma^2\tilde{\mathcal{B}}^{20}(\tau)+\gamma\omega^2\tilde{\mathcal{B}}^{21}(\tau).
\end{align}
whereby we implicitly defined the matrices $\tilde{\mathcal{B}}^0(\tau),\tilde{\mathcal{B}}^1(\tau),\tilde{\mathcal{B}}^{20}(\tau),\tilde{\mathcal{B}}^{21}(\tau)$. We can now write the function $\mathcal{B}(\tau)$ as follows
\begin{align}
\mathcal{B}(\tau)=\gamma \tilde{\mathcal{B}}^0(\tau)+\alpha\gamma\omega\tilde{\mathcal{B}}^1(\tau)+\alpha^2\gamma^2\tilde{\mathcal{B}}^{20}(\tau)+\alpha^2\gamma\omega^2\tilde{\mathcal{B}}^{21}(\tau).\label{tilde_B}
\end{align}
A direct substitution of \eqref{tilde_B} into \eqref{ODE_for_A} allows us to express the function $\mathcal{A}(\tau)$ as
\begin{align}
\mathcal{A}(\tau)=&\omega\left(r_i-r_j\right)\tau+\gamma Tr\left[\Omega\Omega^{\top}\int_{0}^{\tau}{\tilde{\mathcal{B}}^0(u)du}\right]+\alpha\gamma\omega Tr\left[\Omega\Omega^{\top}\int_{0}^{\tau}{\tilde{\mathcal{B}}^1(u)du}\right]\nonumber\\
&+\alpha^2\gamma^2Tr\left[\Omega\Omega^{\top}\int_{0}^{\tau}{\tilde{\mathcal{B}}^{20}(u)du}\right]+\alpha^2\gamma\omega^2Tr\left[\Omega\Omega^{\top}\int_{0}^{\tau}{\tilde{\mathcal{B}}^{21}(u)du}\right]\nonumber\\
=&\omega\left(r_i-r_j\right)\tau+\gamma\tilde{\mathcal{A}}^0(\tau)+\alpha\gamma\omega\tilde{\mathcal{A}}^1(\tau)\nonumber\\
&+\alpha^2\gamma^2\tilde{\mathcal{A}}^{20}(\tau)+\alpha^2\gamma\omega^2\tilde{\mathcal{A}}^{21}(\tau),
\end{align}
having again implicitly defined the functions $\tilde{\mathcal{A}}^0(\tau),\tilde{\mathcal{A}}^1(\tau),\tilde{\mathcal{A}}^{20}(\tau),\tilde{\mathcal{A}}^{21}(\tau)$.
We consider now the pricing in terms of the Fourier transform, i.e. $\omega=i\lambda$, as in \eqref{price2}. Let $\mathcal{Z}$ denote the strip of regularity of the payoff. A Taylor-McLaurin expansion w.r.t. $\alpha$ gives the following:
\begin{align}\label{price_expansion}
C(S(t),K,\tau)\approx&\frac{e^{-r_i\tau}}{2\pi}\int_{\mathcal{Z}}{e^{\mathtt{i}\lambda\left(r_i-r_j\right)\tau+\mathtt{i}\lambda x+\gamma\left(\tilde{\mathcal{A}}^0(\tau)+Tr\left[\tilde{\mathcal{B}}^0(\tau)\Sigma\right]\right)}\Phi(\lambda)d\lambda}\nonumber\\
&+\alpha\left(\tilde{\mathcal{A}}^1(\tau)+Tr\left[\tilde{\mathcal{B}}^1(\tau)\Sigma\right]\right)\nonumber\\
&\frac{e^{-r_i\tau}}{2\pi}\int_{\mathcal{Z}}{\gamma\mathtt{i}\lambda e^{\mathtt{i}\lambda\left(r_i-r_j\right)\tau+\mathtt{i}\lambda x+\gamma\left(\tilde{\mathcal{A}}^0(\tau)+Tr\left[\tilde{\mathcal{B}}^0(\tau)\Sigma\right]\right)}\Phi(\lambda)d\lambda}\nonumber\\
&+\alpha^2\left(\tilde{\mathcal{A}}^{20}(\tau)+Tr\left[\tilde{\mathcal{B}}^{20}(\tau)\Sigma\right]\right)\nonumber\\
&\times\frac{e^{-r_i\tau}}{2\pi}\int_{\mathcal{Z}}{\gamma^2 e^{\mathtt{i}\lambda\left(r_i-r_j\right)\tau+\mathtt{i}\lambda x+\gamma\left(\tilde{\mathcal{A}}^0(\tau)+Tr\left[\tilde{\mathcal{B}}^0(\tau)\Sigma\right]\right)}\Phi(\lambda)d\lambda}\nonumber\\
&+\alpha^2\left(\tilde{\mathcal{A}}^{21}(\tau)+Tr\left[\tilde{\mathcal{B}}^{21}(\tau)\Sigma\right]\right)\nonumber\\
&\times\frac{e^{-r_i\tau}}{2\pi}\int_{\mathcal{Z}}{\gamma\mathtt{i}\lambda^2 e^{\mathtt{i}\lambda\left(r_i-r_j\right)\tau+\mathtt{i}\lambda x+\gamma\left(\tilde{\mathcal{A}}^0(\tau)+Tr\left[\tilde{\mathcal{B}}^0(\tau)\Sigma\right]\right)}\Phi(\lambda)d\lambda}\nonumber\\
&+\frac{\alpha^2}{2}\left(\tilde{\mathcal{A}}^{1}(\tau)+Tr\left[\tilde{\mathcal{B}}^{1}(\tau)\Sigma\right]\right)^2\nonumber\\
&\times\frac{e^{-r_i\tau}}{2\pi}\int_{\mathcal{Z}}{\gamma^2\mathtt{i}^2\lambda^2 e^{\mathtt{i}\lambda\left(r_i-r_j\right)\tau+\mathtt{i}\lambda x+\gamma\left(\tilde{\mathcal{A}}^0(\tau)+Tr\left[\tilde{\mathcal{B}}^0(\tau)\Sigma\right]\right)}\Phi(\lambda)d\lambda}.\nonumber
\end{align}
Recall now from \eqref{int_var_2} the definition of the integrated Black-Scholes variance. In the previous formula in the first term we recognise the Black Scholes price in terms of the characteristic function when the integrated variance is $v=\sigma^2\tau$:
\begin{align}
C_{B\&S}\left(S(t),K,\sigma,\tau\right)=\frac{e^{-r_i\tau}}{2\pi}\int_{\mathcal{Z}}{e^{\mathtt{i}\lambda\left(r_i-r_j\right)\tau+i\lambda x+\frac{(i\lambda)^2-i\lambda}{2}v}\Phi(\lambda)d\lambda},
\end{align}
so that the price expansion is of the form:
\begin{align}
C(S(t),K,\tau)\approx &C_{B\&S}\left(S(t),K,\sigma,\tau\right)\nonumber\\
&+\alpha\left(\tilde{\mathcal{A}}^1(\tau)+Tr\left[\tilde{\mathcal{B}}^1(\tau)\Sigma\right]\right)\partial^2_{xv}C_{B\&S}\left(S(t),K,\sigma,\tau\right)\nonumber\\
&+\alpha^2\left(\tilde{\mathcal{A}}^{20}(\tau)+Tr\left[\tilde{\mathcal{B}}^{20}(\tau)\Sigma\right]\right)\partial^2_{v^2}C_{B\&S}\left(S(t),K,\sigma,\tau\right)\nonumber\\
&+\alpha^2\left(\tilde{\mathcal{A}}^{21}(\tau)+Tr\left[\tilde{\mathcal{B}}^{21}(\tau)\Sigma\right]\right)\partial^3_{x^2v}C_{B\&S}\left(S(t),K,\sigma,\tau\right)\nonumber\\
&+\frac{\alpha^2}{2}\left(\tilde{\mathcal{A}}^{1}(\tau)+Tr\left[\tilde{\mathcal{B}}^{1}(\tau)\Sigma\right]\right)^2\partial^4_{x^2v^2}C_{B\&S}\left(S(t),K,\sigma,\tau\right),\label{price_expansion_2}
\end{align}
which completes the proof.

\subsection{Proof of Proposition \ref{exp_Wis2}}
We follow the procedure in \cite{dafgra11}. We suppose an expansion for the integrated implied variance of the form $v=\sigma_{imp}^2\tau=\zeta_0+\alpha\zeta_1+\alpha^2\zeta_2$ and we consider the Black Scholes formula as a function of the integrated implied variance and the log exchange rate $x=\log S$: $C_{B\&S}(S(t),K,\sigma,\tau)=C_{B\&S}(x(t),K,\sigma^2_{imp}\tau,\tau)$. A Taylor-McLaurin expansion gives us the following:
\begin{align}
C_{B\&S}(x(t),K,\sigma^2_{imp}\tau,\tau)&=C_{B\&S}(x(t),K,\zeta_0,\tau)+\alpha\zeta_1\partial_{v}C_{B\&S}(x(t),K,\zeta_0,\tau)\nonumber\\
&+\frac{\alpha^2}{2}\left(2\zeta_2\partial_v C_{B\&S}(x(t),K,\zeta_0,\tau)\right.\nonumber\\
&\left.+\zeta_1^2\partial^2_{v^2}C_{B\&S}(x(t),K,\zeta_0,\tau)\right).
\end{align}
By comparing this with the price expansion \eqref{price_expansion_2} we deduce that the coefficients must be of the form:
\begin{align}
\zeta_0&=v_0\\
\zeta_1&=\frac{\left(\tilde{\mathcal{A}}^1(\tau)+Tr\left[\tilde{\mathcal{B}}^1(\tau)\Sigma\right]\right)\partial^2_{xv}C_{B\&S}}{\partial_v C_{B\&S}}\label{vol_e_1}\\
\zeta_2&=\frac{-\zeta_1^2\partial^2_{\xi^2} C_{B\&S}+2\left(\tilde{\mathcal{A}}^{20}(\tau)+Tr\left[\tilde{\mathcal{B}}^{20}(\tau)\Sigma\right]\right)\partial^2_{v^2}C_{B\&S}}{2\partial_v C_{B\&S}}\nonumber\\
&+\frac{2\left(\tilde{\mathcal{A}}^{21}(\tau)+Tr\left[\tilde{\mathcal{B}}^{21}(\tau)\Sigma\right]\right)\partial^3_{x^2v}C_{B\&S}}{2\partial_v C_{B\&S}}\nonumber\\
&+\frac{\left(\tilde{\mathcal{A}}^{1}(\tau)+Tr\left[\tilde{\mathcal{B}}^{1}(\tau)\Sigma\right]\right)^2\partial^4_{x^2v^2}C_{B\&S}}{2\partial_v C_{B\&S}},\label{vol_e_3}
\end{align}
where the Black Scholes formula $C_{B\&S}(x(t),K,\sigma_{imp}^2\tau,\tau)$ is evaluated at the point $\left(x,K,v_0,\tau\right)$. In order to find the values of $\zeta_1,\zeta_2$, we differentiate \eqref{B0}-\eqref{B21} thus obtaining the following ODE's:
\begin{align}
\frac{\partial}{\partial \tau}\tilde{\mathcal{B}}^0&=\tilde{\mathcal{B}}^0(\tau)\tilde{M}+\tilde{M}^{\top}\tilde{\mathcal{B}}^0(\tau)+\left(A_i-A_j\right)^2,\\
\frac{\partial}{\partial \tau}\tilde{\mathcal{B}}^1&=\tilde{\mathcal{B}}^1(\tau)\tilde{M}+\tilde{M}^{\top}\tilde{\mathcal{B}}^1(\tau)\nonumber\\
&+\tilde{\mathcal{B}}^0(\tau)Q^{\top}R\left(A_i-A_j\right)+\left(A_i-A_j\right)R^{\top}Q\tilde{\mathcal{B}}^0(\tau),\\
\frac{\partial}{\partial\tau}\tilde{\mathcal{B}}^{20}&=\tilde{\mathcal{B}}^{20}(\tau)\tilde{M}+\tilde{M}^{\top}\tilde{\mathcal{B}}^{20}(\tau)+2\tilde{\mathcal{B}}^{0}(\tau)Q^{\top}Q\tilde{\mathcal{B}}^{0}(\tau),\\
\frac{\partial}{\partial\tau}\tilde{\mathcal{B}}^{21}&=\tilde{\mathcal{B}}^{21}(\tau)\tilde{M}+\tilde{M}^{\top}\tilde{\mathcal{B}}^{21}(\tau)\nonumber\\
&+\tilde{\mathcal{B}}^1(\tau)Q^{\top}R\left(A_i-A_j\right)+\left(A_i-A_j\right)R^{\top}Q\tilde{\mathcal{B}}^1(\tau).
\end{align}
We consider a Taylor-McLaurin expansion in terms of $\tau$
\begin{align}
\tilde{\mathcal{B}}^0(\tau)&\approx\left(A_i-A_j\right)^2\tau+\frac{\tau^2}{2}\left(\left(A_i-A_j\right)^2\tilde{M}+\tilde{M}^{\top}\left(A_i-A_j\right)^2\right)\\
\tilde{\mathcal{B}}^1(\tau)&\approx\frac{\tau^2}{2}\left[\left(A_i-A_j\right)^2Q^{\top}R\left(A_i-A_j\right)+\left(A_i-A_j\right)R^{\top}Q\left(A_i-A_j\right)^2\right]\nonumber\\
&+\frac{\tau^3}{6}\Big(\left[\left(A_i-A_j\right)^2Q^{\top}R\left(A_i-A_j\right)+\left(A_i-A_j\right)R^{\top}Q\left(A_i-A_j\right)^2\right]\tilde{M}\Big.\nonumber\\
&+\Big.\tilde{M}^{\top}\left[\left(A_i-A_j\right)^2Q^{\top}R\left(A_i-A_j\right)+\left(A_i-A_j\right)R^{\top}Q\left(A_i-A_j\right)^2\right]\Big)\nonumber\\
&+\frac{\tau^3}{6}\Big(\left[\left(A_i-A_j\right)^2\tilde{M}+\tilde{M}^{\top}\left(A_i-A_j\right)^2\right]Q^{\top}R\left(A_i-A_j\right)\Big.\nonumber\\
&+\Big.\left(A_i-A_j\right)R^{\top}Q\left[\left(A_i-A_j\right)^2\tilde{M}+\tilde{M}^{\top}\left(A_i-A_j\right)^2\right]\Big)\\
\tilde{\mathcal{B}}^{20}(\tau)&\approx\frac{\tau^3}{6}4\left(A_i-A_j\right)^2Q^{\top}Q\left(A_i-A_j\right)^2\\
\tilde{\mathcal{B}}^{21}(\tau)&\approx\frac{\tau^3}{6}\Big(\left[\left(A_i-A_j\right)^2Q^{\top}R\left(A_i-A_j\right)+\left(A_i-A_j\right)R^{\top}Q\left(A_i-A_j\right)^2\right]\Big.\nonumber\\
&\times Q^{\top}R\left(A_i-A_j\right)+\left(A_i-A_j\right)R^{\top}Q\times\nonumber\\
&\Big.\left[\left(A_i-A_j\right)^2Q^{\top}R\left(A_i-A_j\right)+\left(A_i-A_j\right)R^{\top}Q\left(A_i-A_j\right)^2\right]\Big).
\end{align}

Noticing from \eqref{mathcal_A0}-\eqref{mathcal_A21} that $\tilde{\mathcal{A}}^{i}(\tau)$ are one order in $\tau$ higher than the corresponding $\tilde{\mathcal{B}}^{i}(\tau)$, the following approximations hold:
\begin{align}
\tilde{\mathcal{A}}^0(\tau)+Tr\left[\tilde{\mathcal{B}}^0(\tau)\Sigma\right]&=Tr\left[\left(A_i-A_j\right)^2\Sigma\right]\tau+o(\tau)\label{start}\\
\tilde{\mathcal{A}}^1(\tau)+Tr\left[\tilde{\mathcal{B}}^1(\tau)\Sigma\right]&=Tr\left[\left(A_i-A_j\right)^2Q^{\top}R\left(A_i-A_j\right)\Sigma\right]\tau^2+o(\tau^2)\\
\tilde{\mathcal{A}}^{20}(\tau)+Tr\left[\tilde{\mathcal{B}}^{20}(\tau)\Sigma\right]&=\frac{2}{3}Tr\left[\left(A_i-A_j\right)^2Q^{\top}Q\left(A_i-A_j\right)^2\right]\tau^3+o(\tau^3)\\
\tilde{\mathcal{A}}^{21}(\tau)+Tr\left[\tilde{\mathcal{B}}^{21}(\tau)\Sigma\right]&=Tr\Big[\left(\left(A_i-A_j\right)^2Q^{\top}R\left(A_i-A_j\right)\right.\Big.\nonumber\\
&+\Big.\left.\left(A_i-A_j\right)R^{\top}Q\left(A_i-A_j\right)^2\right)\Big.\nonumber\\
&\Big.\times Q^{\top}R\left(A_i-A_j\right)\Sigma\Big]\frac{\tau^3}{3}+o(\tau^3).
\end{align}
We introduce two variables: the log-moneyness $m_f=\log\left(\frac{S^{i,j}(t)e^{(r_i-r_j)\tau}}{K}\right)$ and the variance $V=Tr\left[\left(A_i-A_j\right)\Sigma(t)\left(A_i-A_j\right)\right]\tau$. 
Then, from \cite{lewis2000}, we consider the following ratios among the derivatives of the Black-Scholes formula:
\begin{align}
\frac{\partial^2_{x,v}C_{B\&S}\left(x,K,V,\tau\right)}{\partial_{v}C_{B\&S}\left(x,K,V,\tau\right)}&=\frac{1}{2}+\frac{m_f}{V};\\
\frac{\partial^2_{v^2}C_{B\&S}\left(x,K,V,\tau\right)}{\partial_{v}C_{B\&S}\left(x,K,V,\tau\right)}&=\frac{m_f^2}{2V^2}-\frac{1}{2V}-\frac{1}{8};\\
\frac{\partial^3_{x^2,v}C_{B\&S}\left(x,K,V,\tau\right)}{\partial_{v}C_{B\&S}\left(x,K,V,\tau\right)}&=\frac{1}{4}+\frac{m_f-1}{V}+\frac{m_f^2}{V^2};\\
\frac{\partial^4_{x^2,v^2}C_{B\&S}\left(x,K,V,\tau\right)}{\partial_{v}C_{B\&S}\left(x,K,V,\tau\right)}&=\frac{m_f^4}{2V^4}+\frac{m_f^2\left(m_f-1\right)}{2V^3}.\label{end}
\end{align}
Upon substitution of \eqref{start}-\eqref{end} into \eqref{vol_e_1}-\eqref{vol_e_3} we obtain the values for $\zeta_i,i=0,1,2$ allowing us to express the expansion of the implied volatility.
\begin{align}
\zeta_0&=Tr\left[\left(A_i-A_j\right)\Sigma(t)\left(A_i-A_j\right)\right]\tau,\\
\zeta_1&=\frac{Tr\left[\left(A_i-A_j\right)^2Q^{\top}R\left(A_i-A_j\right)\Sigma(t)\right]m_f\tau}{Tr\left[\left(A_i-A_j\right)\Sigma(t)\left(A_i-A_j\right)\right]},\\
\zeta_3&=\frac{m_f^2}{Tr\left[\left(A_i-A_j\right)\Sigma(t)\left(A_i-A_j\right)\right]^2}\tau\Bigg[\frac{1}{3}Tr\left[\left(A_i-A_j\right)^2Q^\top Q\left(A_i-A_j\right)^2\Sigma(t)\right]\Bigg.\nonumber\\
&+\frac{1}{3}Tr\Big[\left[\left(A_i-A_j\right)^2Q^{\top}R\left(A_i-A_j\right)+\left(A_i-A_j\right)R^{\top}Q\left(A_i-A_j\right)^2\right]\Big.\nonumber\\
&\Bigg.\Big.\times Q^\top R\left(A_i-A_j\right)\Sigma(t)\Big]-\frac{5}{4}\frac{Tr\left[\left(A_i-A_j\right)^2Q^{\top}R\left(A_i-A_j\right)\Sigma(t)\right]^2}{Tr\left[\left(A_i-A_j\right)\Sigma(t)\left(A_i-A_j\right)\right]}\Bigg].
\end{align}
By plugging these expressions we obtain the result.

\newpage

\section{Images and Tables}

\begin{table}[H]
	\centering
		\begin{tabular}{crcr}
Parameter&Value&Parameter&Value\\
\hline\\
$\Sigma(t)(1,1)$&0.1688&$A_{us}(1,1)$&0.7764\\
$\Sigma(t)(1,2)$&0.1708&$A_{us}(1,2)$&0.4837\\
$\Sigma(t)(2,2)$&0.3169&$A_{us}(2,2)$&0.9639\\
$M(1,1)$&-0.5213&$A_{eur}(1,1)$&0.6679\\
$M(1,2)$&-0.3382&$A_{eur}(1,2)$&0.6277\\
$M(2,1)$&-0.4940&$A_{eur}(2,2)$&0.8520\\
$M(2,2)$&-0.4389&$h_{us}$&-0.2218\\
$Q(1,1)$&0.2184&$h_{eur}$&-0.1862\\
$Q(1,2)$&0.0957&$H_{us}(1,1)$&0.2725\\
$Q(2,1)$&0.2483&$H_{us}(1,2)$&0.0804\\
$Q(2,2)$&0.3681&$H_{us}(2,2)$&0.4726\\
$R(1,1)$&-0.5417&$H_{eur}(1,1)$&0.1841\\
$R(1,2)$&0.1899&$H_{eur}(1,2)$&0.0155\\
$R(2,1)$&-0.1170&$H_{eur}(2,2)$&0.4761\\
$R(2,2)$&-0.4834&$\beta$&3.1442\\
\hline
		\end{tabular}
	\caption{This table reports the results of the calibration of the FX-IR Wishart hybrid model}
	\label{tab:paramsWisIRFX}
\end{table}

\begin{figure}[H]
  \centering                           \subfloat{\label{fig:fig111}\includegraphics[scale=0.50]{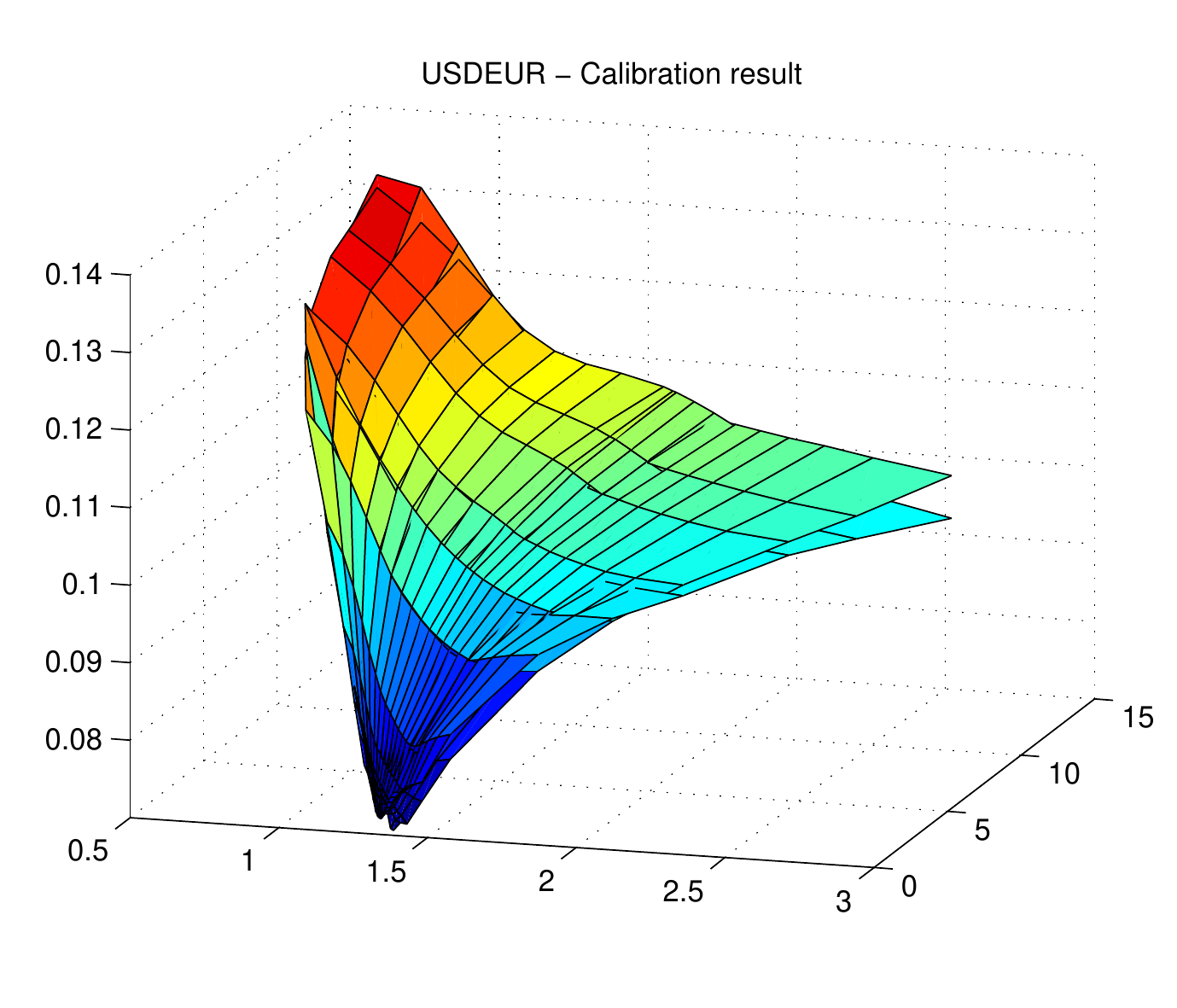}} \
  \subfloat{\label{fig:fig112}\includegraphics[scale=0.50]{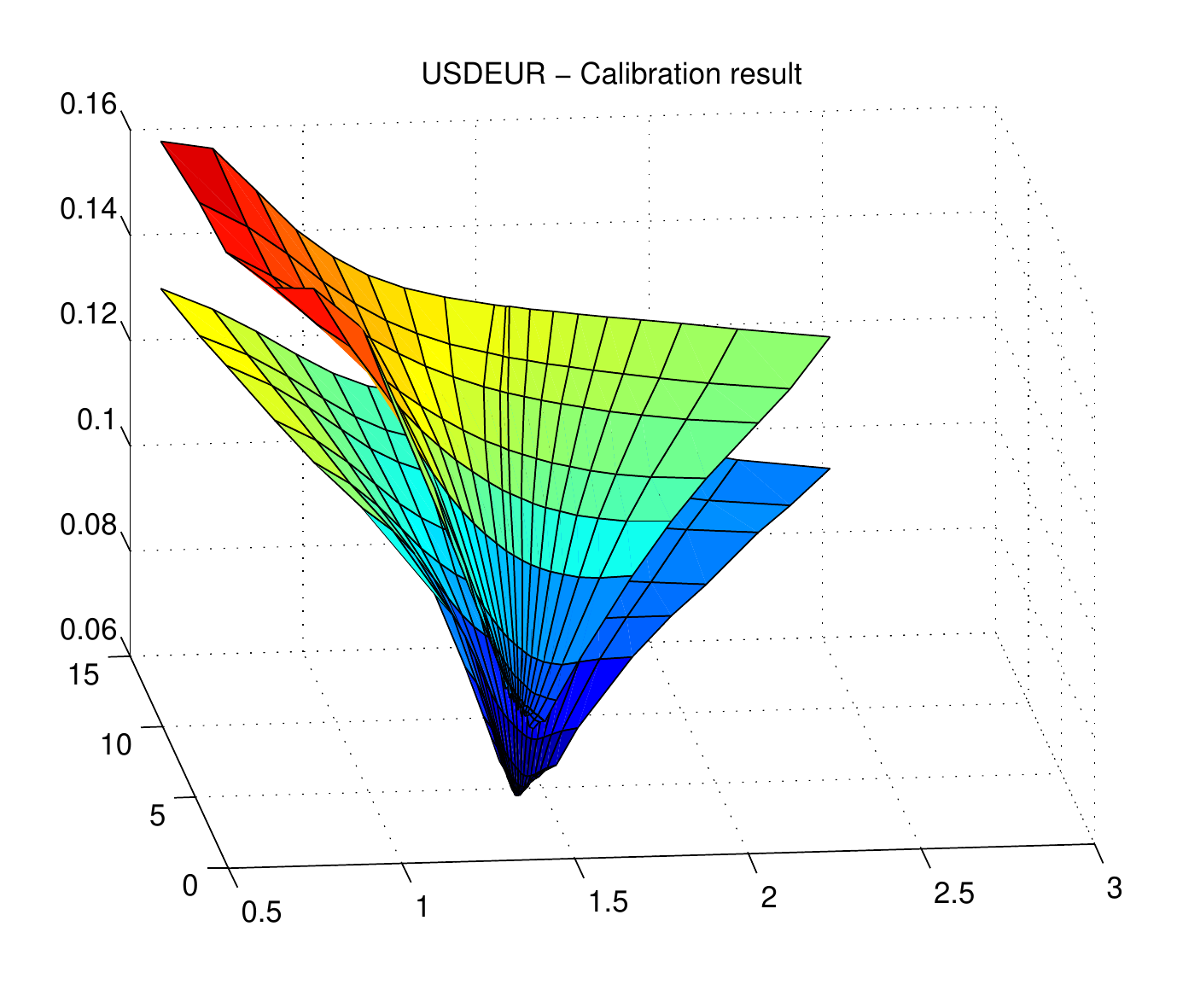}}
    \subfloat{\label{fig:fig113}\includegraphics[scale=0.50]{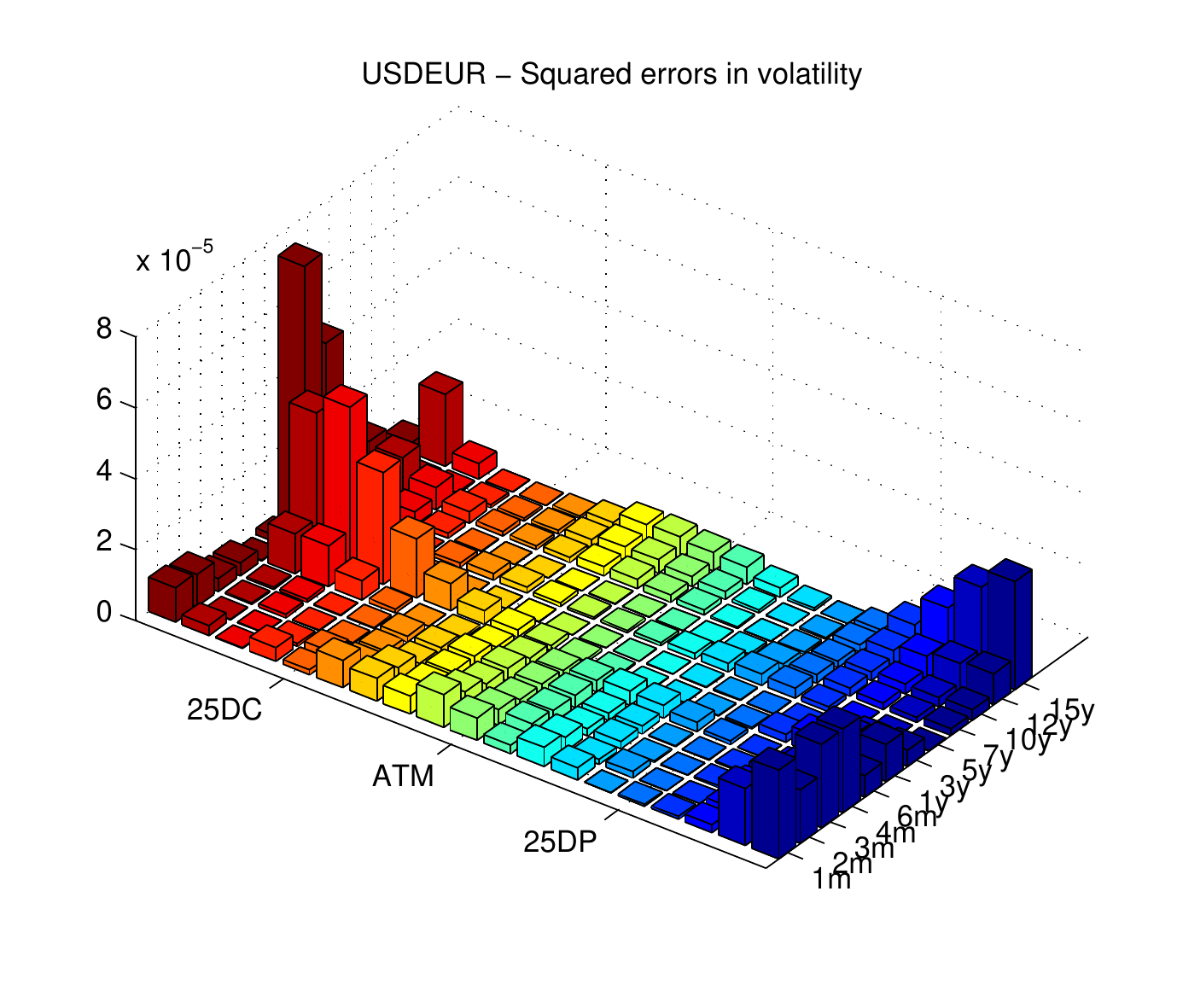}} \                                  
\caption{We consider a joint calibration of the implied volatility surface of USDEUR and of the two yield curves of USD and EUR simultaneously. The figure on the top compares the implied volatility surface to observed market data. As the fit is very good and the two surfaces are almost indistinguishable, on the bottom left figure we shift upwards the model implied volatility by multiplying each point by 1.1 so as to encrease readability. On the bottom right we report squared errors in implied volatility for each point of the surface.}
  \label{fig:stoch_rate_1}
\end{figure}

\begin{figure}[H]
  \centering                           \subfloat{\label{fig:fig121}\includegraphics[scale=0.50]{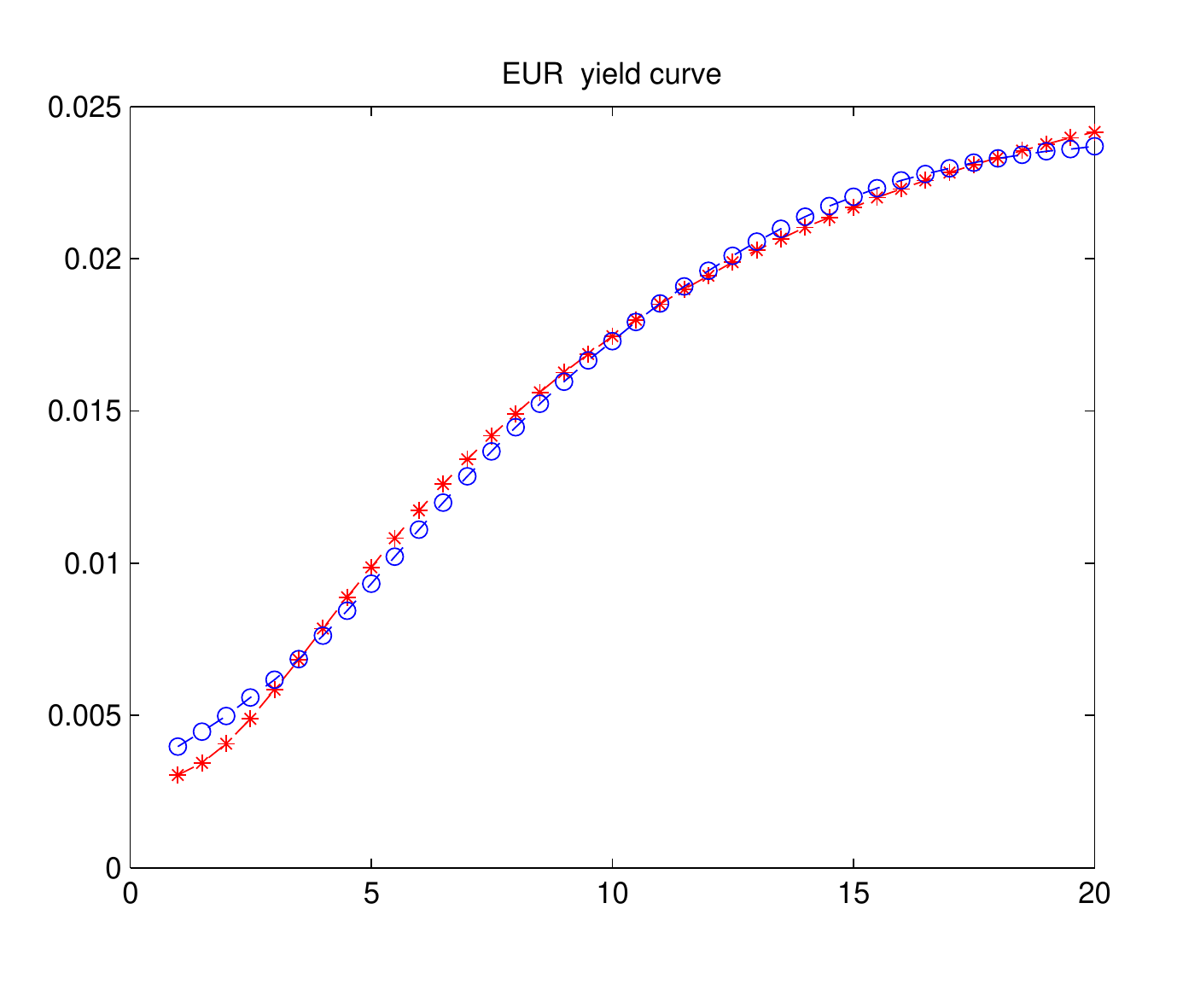}} \
\subfloat{\label{fig:fig122}\includegraphics[scale=0.50]{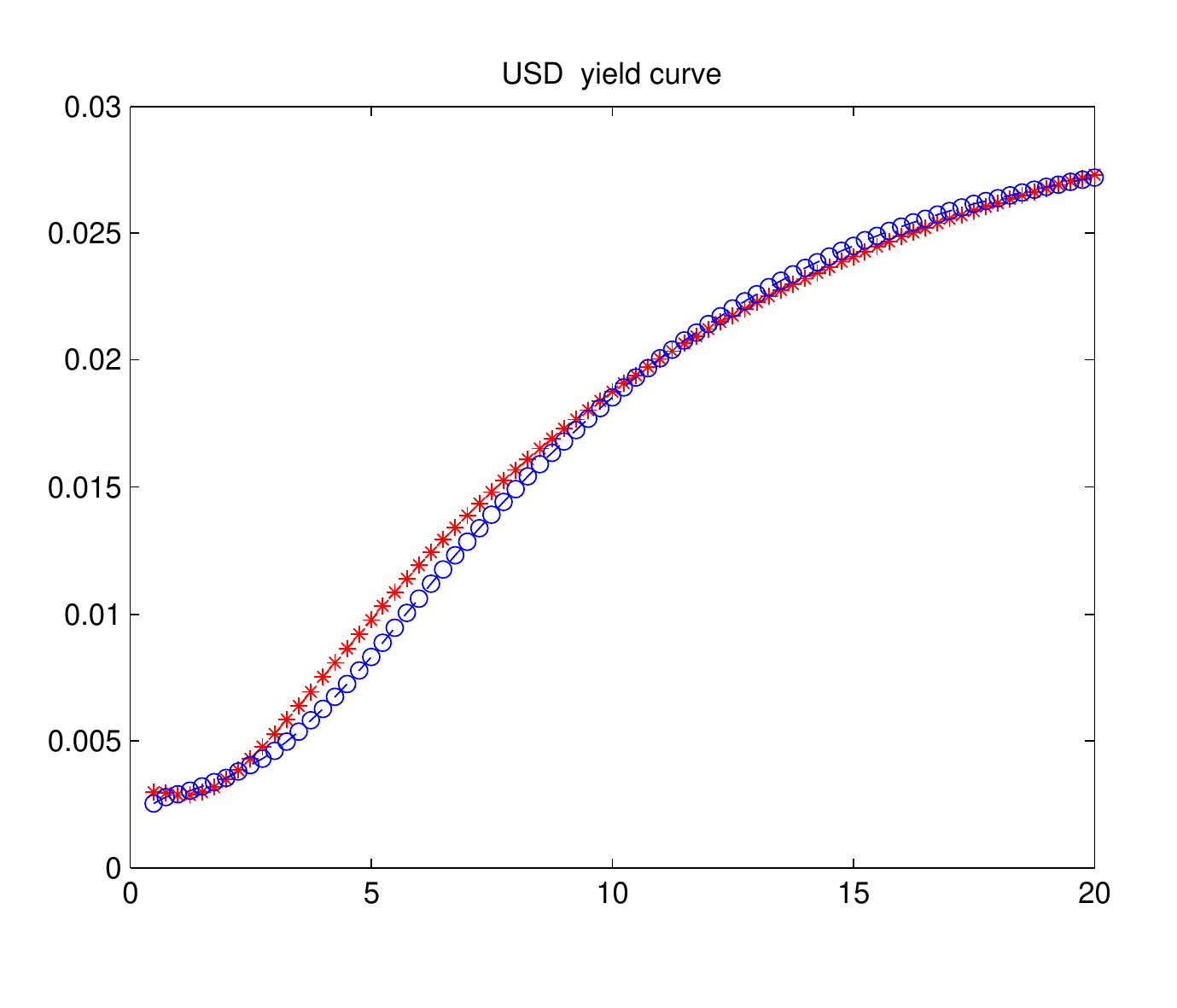}} \           
\caption{We consider a joint calibration of the implied volatility surface of USDEUR and of the two yield curves of USD and EUR simultaneously. Stars and circles denote model and market implied yields respectively.}
  \label{fig:stoch_rate_2}
\end{figure}

\appendix

\bibliography{biblio}
\bibliographystyle{apalike}
\bibliographystyle{apa}

\end{document}